\documentclass[11pt]{article}

\usepackage{amsmath,amsfonts,amssymb,nccmath}

\usepackage{cancel}
\usepackage{tikz-cd}
\usetikzlibrary{arrows}

\usepackage{float}

\usepackage[numbers]{natbib}
\usepackage{graphicx}

\usepackage{hyperref}
\hypersetup{
   colorlinks   =  true,
    linkcolor    = cyan,
    citecolor    = red,
     urlcolor	=magenta,     
}

\usepackage{amsthm}
\theoremstyle{definition} 

\theoremstyle{theorem} 

\newtheorem{proposition}{Proposition}
\newtheorem{corollary}{Corollary}

\newcommand{\blue}[1]{{\color{blue}#1}}

\def\be{\begin{equation}}
\def\ee{\end{equation}}
\def\bea{\begin{eqnarray}}
\def\eea{\end{eqnarray}}

\begin{document}

\hfill\today

\begin{center}
\baselineskip 24 pt {\LARGE \bf  
Exact closed-form solution of a modified SIR model}
\end{center}

\begin{center}

{\sc Angel Ballesteros$^1$, Alfonso Blasco$^1$, Ivan Gutierrez-Sagredo$^{1,2}$}

\medskip

{$^1$Departamento de F\'isica, Universidad de Burgos, 
09001 Burgos, Spain}

{$^2$Departamento de Matem\'aticas y Computaci\'on, Universidad de Burgos, 
09001 Burgos, Spain}

 \medskip
 
e-mail: {\href{mailto:angelb@ubu.es}{angelb@ubu.es}, \href{mailto:ablasco@ubu.es}{ablasco@ubu.es}, \href{mailto:igsagredo@ubu.es}{igsagredo@ubu.es}}

\end{center}

\medskip

\begin{abstract}

The exact analytical solution in closed form of a modified SIR system where recovered individuals are removed from the population is presented. In this dynamical system the populations $S(t)$ and $R(t)$ of susceptible and recovered individuals are found to be generalized logistic functions, while infective ones $I(t)$ are given by a generalized logistic function times an exponential, all of them with the same characteristic time. The dynamics of this modified SIR system is analyzed and the exact computation of some epidemiologically relevant quantities is performed. The main differences between this modified SIR model and original SIR one are presented and explained in terms of the zeroes of their respective conserved quantities. Moreover, it is shown that the modified SIR model with time-dependent transmission rate can be also solved in closed form for certain realistic transmission rate functions.

\end{abstract}

\bigskip

\noindent Keywords: epidemics, compartmental models, dynamical systems, exact solution, SIR model, time-dependent transmission rate


\section{Introduction}

The so-called deterministic compartmental dynamical systems are the simplest amongst the models of epidemiogical dynamics, and a large number of them have been recently considered in relation to the COVID-19 pandemic (see, for example, \cite{MB2020science,VLABHPRYV2020,CG2020size} and references therein).  The study of these models typically relies on techniques from dynamical systems theory and numerical studies but, despite these techniques allow a deep understanding of their associated dynamics, the simplicity and accurateness provided by exact simple solutions are indeed helpful both from the mathematical and the epidemiological perspectives (see for instance~\cite{Hethcote1973,Bailey1975,Nucci2004} for the exact solutions of the two-dimensional SIS (susceptible-infective-susceptible) model).

Among three-dimensional models, the well-known SIR (susceptible-infective-recovered) system 
\begin{equation}
\label{eq:SIR}
\dot{S} = -\beta\,S\,I, \qquad\qquad\qquad \dot{I} = \beta\,S\,I-\alpha \,I, \qquad\qquad\qquad \dot R = \alpha I , \\
\end{equation}
proposed by Kermack and McKendrick \cite{KM1927sir} is probably the best known one. Despite its apparent simplicity, it has been succesfully used to predict relevant features of the dynamics of a number of epidemics, including the actual COVID-19 pandemic \cite{Buonomo2020,Postnikov2020sir}. Therefore, the study of exact solutions for this system has been faced from several perspectives: 
Painlev\'e analysis and Lie symmetries~\cite{LA2004}, parametric-form solutions~\cite{Harko2014}, asymptotic approximants~\cite{Barlow2020}, Hamiltonian structures~\cite{BBG2020hamiltonianepidemics} and time reparametrization~\cite{Cadoni2020}. 
Nevertheless, all these approaches lead to solutions which are either perturbative or given in terms of one implicit (or inverse) function. In this sense it can be said that the system \eqref{eq:SIR} does not admit an `exact analytic solution in closed form', {\em i.e.} a solution that can be expressed in terms of a finite number of ordinary operations among elementary functions.


In contradistinction to this fact,  in this paper we present the exact analytical solution in closed form of the modified SIR system~\cite{ Brauer1990}
\begin{equation}
\label{eq:modifiedSIR}
\dot{S} = -\dfrac{\beta\,S\,I}{S+I}, \qquad\qquad\qquad \dot{I} = \dfrac{\beta\,S\,I}{S+I}-\alpha \,I , \qquad\qquad\qquad \dot{R} = \alpha\,I ,
\end{equation}
where $\alpha, \beta \in \mathbb R_+$. We show that, surprisingly enough, the general solution of this modified SIR system is given in terms of generalized logistic and exponential functions, namely
\begin{equation}
\label{eq:sol}
S(t) = S_0 \left( \frac{S_0+I_0}{S_0+ I_0 e^{t/\tau}}  \right)^{\beta \tau} , \qquad I(t) = I_0 \left( \frac{S_0 + I_0}{S_0 + I_0 e^{t/\tau}} \right)^{\beta \tau} e^{t/\tau}, \qquad R(t) = 1 - \frac{(S_0 + I_0)^{\beta \tau}}{(S_0+I_0 e^{t/\tau})^{\beta \tau -1}},
\end{equation}
where $\tau=(\beta-\alpha)^{-1}$. This is, to the best of our knowledge, the first exact solution of a three-dimensional compartmental epidemiological model in closed form. We will also analyze the modified SIR model from a dynamical systems perspective and show that, in some range of the model parameters $(\alpha, \beta)$, the dynamics of \eqref{eq:modifiedSIR} is actually quite close to the one of the SIR system \eqref{eq:SIR}. 

We recall that the modified SIR model~\eqref{eq:modifiedSIR} has been proposed \cite{Brauer1990,Gumral1993} as the appropriate generalization of the SIR model when the recovered individuals are removed from the population. Therefore, although this model is not directly applicable to the COVID-19 pandemic, it will be certainly meaningful for the study of diseases with a very high death rate (the so-called `fatal diseases' like, for instance, the Feline Infectious Peritonitis (FIP), Spongiform Encephalopathy (BSE), Leishmaniasis, Rabbit Haemorrhagic Disease, and the Highly Pathogenic Avian Influenza (H5N1))~\cite{Keeling}) or also in cases when very prolongued quarantines are prescribed.

Moreover, the modified SIR model with a time-dependent transmission rate $\beta(t)$ will be also shown to be exactly solvable in closed form for certain realistic $\beta(t)$ functions, and its solutions will be compared with the constant rate system.
It is worth stressing that these time-dependent transmission rate models are well-known to be the appropriate ones for the analysis of modifications of the transmission rate of a given disease, that can be due -for instance- to seasonal effects, to changes in host behaviour or immunity or to modifications in the abundance of vectors/pathogens (see for instance~\cite{Keeling, BacaerGomes, Ponciano, Liu, Pollicott, Mummert} and references therein). As a consequence, $\beta(t)$ models are indeed of the outmost interest for the planning of non-pharmacological control strategies and, again, the existence of exact closed-form solutions for any of them was lacking.

The structure of the paper is the following. In the next Section we derive the exact solution~\eqref{eq:sol} by making use of the fact that any epidemiological three-dimensional model has a conserved quantity, which in turn is straightforwardly derived from the the more general result (recently proved in \cite{BBG2020hamiltonianepidemics}) stating that any three-dimensional compartmental epidemiological  model is a generalized Hamiltonian system. Moreover, the conserved quantity turns out to be just the Casimir of the Poisson algebra of the underlying Hamiltonian structure. In Section 3 we present the analysis of the modified SIR system~\eqref{eq:modifiedSIR} both from a dynamical systems approach and from a Poisson--algebraic point of view,  and we show that the exact solution~\eqref{eq:sol} is helpful in order to obtain some relevant epidemiological quantities in a simple and exact form. The main differences between the SIR and modified SIR dynamical systems are analysed in Section 4,  where we show that these differences can be understood in terms of the the zeroes of their respective conserved quantities, which are again the Casimir functions for both models that are obtained through the formalism presented in \cite{BBG2020hamiltonianepidemics}. Finally, in Section 5  we consider the generalization of the modified SIR system \eqref{eq:modifiedSIR} when the transmission rate $\beta (t)$ varies with time. We shall show two instances of transmission rates for which the model is explicitly integrable in closed form, and the numerical solution of the model with periodic transmission rate will be also presented.

\section{Exact solution of the modified SIR model}

In order to find the exact solution of \eqref{eq:modifiedSIR} we make use of the following recent result (see \cite{BBG2020hamiltonianepidemics} for details).  

\begin{proposition} \cite{BBG2020hamiltonianepidemics}
\label{prop:EH}
Every epidemiological compartmental model with constant population is a generalized Hamiltonian system, with Hamiltonian function $\mathcal{H}$ given by the total population.
\end{proposition}

For the system \eqref{eq:modifiedSIR} the generalized Hamiltonian structure is thus explicitly provided by the Hamiltonian function 
\begin{equation}
\label{eq:H}
\mathcal{H}=S+I+R, 
\end{equation}
together with the associated Poisson structure, which is found to be given by the fundamental brackets 
\begin{equation}
\label{eq:Poisson}
\{S,I\}=0, \qquad \{S,R\}=-\dfrac{\beta\,S\,I}{S+I}, \qquad \{I,R\}=\dfrac{\beta\,S\,I}{S+I}-\alpha \,I \, ,
\end{equation}
and leads to the system~\eqref{eq:modifiedSIR} through Hamilton's equations
\begin{equation}
\label{eq:Hamilton_eqs}
\dot S = \{S,\mathcal H \} , \qquad
\dot I = \{I,\mathcal H \} , \qquad
\dot R = \{R,\mathcal H \} .
\end{equation}

Since every three-dimensional Poisson structure has a Casimir function $\mathcal{C}$, i.e. a function $\mathcal C : \mathcal U \subseteq \mathbb R^3 \to \mathbb R$ such that $\{S,\mathcal C \} = \{I,\mathcal C \} = \{R,\mathcal C \}=0$, then $\mathcal{C}$ is a conserved quantity for any generalized Hamiltonian system~\eqref{eq:Hamilton_eqs} defined on such a Poisson manifold. Therefore:
\begin{corollary} 
\label{prop:}
Every three-dimensional epidemiological compartmental model with constant population has a conserved quantity, which is functionally independent of the Hamiltonian function.
\end{corollary}

Note that in case that $\mathcal{H}$~\eqref{eq:H} is functionally dependent of $\mathcal{C}$, the dynamics~\eqref{eq:Hamilton_eqs} would be trivial. For the specific Poisson algebra \eqref{eq:Poisson} the Casimir function is found to be 
\begin{equation}
\label{eq:C}
\mathcal{C}=S^{-\frac{\alpha}{\beta}}(S+I) .
\end{equation}
We can use this Casimir function to restrict the dynamics of \eqref{eq:modifiedSIR} to the symplectic leaf defined by the value of $\mathcal{C}$ given by the  initial conditions $S(0)=S_0$, $I(0)=I_0$, $R(0)=R_0$, namely
\begin{equation}
\label{eq:C0}
\mathcal C_0=S_0^{-\frac{\alpha}{\beta}}(S_0+I_0) \, .
\end{equation}
This can be also used in order to reduce the system \eqref{eq:modifiedSIR} to a nonlinear ODE, since from \eqref{eq:C} and \eqref{eq:C0} we obtain the phase space equation 
\begin{equation}
\label{eq:I_S}
I(S) = (S_0+I_0) \left( \frac{S}{S_0} \right)^{\frac{\alpha}{\beta}} - S ,
\end{equation}
which can be inserted within \eqref{eq:modifiedSIR} in order to get the following nonlinear ODE for the variable $S$:
\begin{equation}
\dot S = - \beta S \left( 1- \frac{S_0^{\alpha/\beta}}{S_0+I_0} S^{1-\alpha/\beta} \right) .
\end{equation}
This ODE suggests the change of variable 
\begin{equation}
\label{eq:AS}
A(t) = S(t)^{1-\alpha/\beta} ,
\end{equation}
thus obtaining 
\begin{equation}
\dot A = - (\beta - \alpha) A \left( 1- \frac{S_0^{\alpha/\beta}}{S_0+I_0} A \right) .
\end{equation}
If we now set 
\begin{equation}
\label{eq:BA}
B(t) = \frac{S_0^{\alpha/\beta}}{S_0+I_0} A (t) ,
\end{equation}
we obtain
\begin{equation}
\label{eq:ODE_B}
\dot B = - (\beta - \alpha) B \left( 1- B \right) .
\end{equation}
The general solution to this ODE is a logistic function with characteristic time $\tau = (\beta-\alpha)^{-1}$, i.e.
\begin{equation}
B(t) = \frac{1}{1+e^{(\beta-\alpha)t +d}}  = \frac{1}{1+e^{t / \tau+d}} .
\end{equation}
The integration constant $d$ is fixed by the initial condition $B(0)=\frac{S_0}{S_0+I_0}$, thus obtaining $e^d=\frac{I_0}{S_0}$. Therefore we can write
\begin{equation}
B(t) = \frac{1}{1+ \frac{I_0}{S_0} e^{t/\tau}} \, .
\end{equation}
Now, inverting the change of variables \eqref{eq:BA} we get
\begin{equation}
A(t) = \frac{(S_0+I_0)S_0^{1/\beta\tau}}{S_0+ I_0 e^{t/\tau}} ,
\end{equation}
and finally, from \eqref{eq:AS}, we obtain
\begin{equation}
\label{eq:St}
S(t) = S_0 \left( \frac{S_0+I_0}{S_0+ I_0 e^{t/\tau}}  \right)^{\beta \tau} .
\end{equation}
From the phase space equation \eqref{eq:I_S} we directly get
\begin{equation}
\label{eq:IS}
I(S) = (S_0+I_0) \left( \frac{S}{S_0} \right)^{\alpha/\beta} - S ,
\end{equation}
and inserting \eqref{eq:St} we are able to obtain $I(t)$ without any further integration. Finally, we have that
\begin{equation}
\label{eq:It}
I(t) = I_0 \left( \frac{S_0 + I_0}{S_0 + I_0 e^{t/\tau}} \right)^{\beta \tau} e^{t/\tau} .
\end{equation}
Note that $I(t)$ is related to $S(t)$ by 
\begin{equation}
\label{eq:ISexp}
I(t) = \frac{I_0}{S_0} S(t) e^{t/\tau} \, ,
\end{equation}
and from the conservation of the total population, we find
\begin{equation}
\label{eq:Rt}
R(t) = 1 -S(t) - I(t) =1 - \frac{(S_0 + I_0)^{\beta \tau}}{(S_0+I_0 e^{t/\tau})^{\beta \tau -1}} .
\end{equation}

Summarizing, equations \eqref{eq:St} shows that the susceptible population follows a generalized logistic function, or Richards' curve, with characteristic time $\tau$ and the relevant constants set to satisfy  that $S(0) = S_0$ and $\lim_{t \to \infty} S(t) = 0$. Moreover, the dynamics of the infective population given by \eqref{eq:It} is essentially this same function multiplied by an exponential with the same characteristic time. This is indeed a very natural dynamics for infective processes and, as we will see in the sequel, this dynamics strongly resembles the one described by the famous SIR model \eqref{eq:SIR}, provided that the range of values for the parameters $\alpha$ and $\beta$ is similar to the one found in actual epidemics. 

\vspace{-0.5cm}

\paragraph{Remark 1}
It is worth stressing that the method here presented is indeed applicable to any three-dimensional compartmental model, provided we are able to find the Casimir function of the associated Poisson structure. Nevertheless, the distinctive feature of the system \eqref{eq:modifiedSIR} is that the resulting ODE admits a closed-form solution. We recall that in \cite{BBG2020hamiltonianepidemics} such Casimir function approach was used in order to find the solution for some epidemiological models in terms of an inverse function. 

\vspace{-0.5cm}

\paragraph{Remark 2}  
Solution~\eqref{eq:sol} suggests that the new variable $y=I/S$ should be worth to be considered, since $y(t)$ is an exponential function. In fact, by taking $(S,y,R)$ as new dynamical variables, the modified SIR system reads
\begin{equation}
\label{eq:MSIRy}
\dot{S} = -\frac{\beta}{1+y}\,S\,y, \qquad\qquad\qquad \dot{y} = (\beta-\alpha) \,y, \qquad\qquad\qquad \dot R = \alpha\,S\,y , \\
\end{equation}
in which the equation for $y$ is linearized, as expected. Note that the original SIR system is written in terms of $y$ as
\begin{equation}
\label{eq:SIRy}
\dot{S} = -\beta\,S^2\,y, \qquad\qquad\qquad \dot{y} = (\beta\,S-\alpha) \,y
+ \beta\,S\,y^2, \qquad\qquad\qquad \dot R = \alpha\,S\,y , \\
\end{equation}
which is quite different from~\eqref{eq:MSIRy} as a dynamical system, as it will be shown in Section 4.

\section{Analysis of the modified SIR dynamics}

In this Section we briefly analyze the main features of the modified SIR system \eqref{eq:modifiedSIR}. Without any loss of generality we can assume that $R_0=0$, so $S_0 + I_0 =1$, and the solution of \eqref{eq:modifiedSIR} reads 
\begin{equation}
\label{eq:solsimply}
S(t) = \frac{S_0}{\left(  S_0+ I_0 e^{t/\tau} \right)^{\beta \tau} }   , \qquad\qquad I(t) =  \frac{I_0 \, e^{t/\tau}}{\left( S_0 + I_0 e^{t/\tau} \right)^{\beta \tau}}  , \qquad\qquad R(t) = 1 - \frac{1}{ \left( S_0+I_0 e^{t/\tau} \right)^{\beta \tau -1}} .
\end{equation}
As we have previously stated, the behavior of $S(t)$ is that of a generalized logistic function while the evolution of the infective population $I(t)$ is given by a generalized logistic function times an exponential. This means that since $\beta \tau > 1$, the logistic term dominates for large times and therefore $\lim_{t \to \infty} I(t) = 0$. However during the first stage of the outbreak the exponential term is the dominating one, and thus the model presents the characteristic infection peak (for appropriate values of the parameters $\alpha$ and $\beta$). The behavior of the functions $S(t)$ and $I(t)$ for different values of $\alpha$ and $\beta$ is shown in Figure \ref{fig:comp_traj}.

A fundamental question to be answered by any epidemiological model is whether, for given values of the parameters, there will be an outbreak. For the modified SIR system we see, simply by evaluating the second equation from \eqref{eq:modifiedSIR} at $t=0$,
 \begin{equation}
\frac{d}{dt}\bigg |_{t=0} I(t) = I_0 \left( \frac{\beta S_0}{S_0 + I_0} - \alpha \right) ,
\end{equation}
and the outbreak will exist if and only if $\beta S_0 > \alpha (S_0 + I_0)$, or equivalently 
\begin{equation}
\label{eq:peak_cond}
\frac{\beta}{\alpha}-1 > \frac{I_0}{S_0} ,
\end{equation}
which in the case $S_0 + I_0 =1$ means that
\begin{equation}
\label{eq:peak_cond2}
 S_0 > \frac{\alpha}{\beta}.
\end{equation}
Obviously, this same result can be obtained by checking the condition for which $I(t)$ has a maximum. Moreover, the analytic solution allows us to exactly determine the time at which the infection peak $t_{\text{peak}}$ is reached, and we obtain
\begin{equation}
\label{eq:tpeak}
t_{\text{peak}} = \tau \log \left( \frac{S_0}{I_0 (\beta \tau - 1)} \right) = \tau \log \left( \frac{S_0}{I_0} \left( \frac{\beta}{\alpha} - 1 \right) \right) \, ,
\end{equation}
which is positive if and only if \eqref{eq:peak_cond} holds. The fraction of infected population at the infection peak reads 
\begin{equation}
\label{eq:Ipeak}
I(t_{\text{peak}}) = \left( \frac{\beta \tau -1}{S_0} \right)^{\beta \tau - 1} \left( \frac{S_0 + I_0}{\beta \tau} \right)^{\beta \tau - 1} = 
S_0 \left( \frac{\beta}{\alpha}-1 \right) \left( \frac{\alpha}{\beta} \left(  1+\frac{I_0}{S_0} \right) \right)^{\beta/(\beta-\alpha)} .
\end{equation}

A relevant epidemiological quantity is the well-known basic reproduction number $\mathcal R_0 $, defined as the average number of secondary cases produced by one infected individual introduced into a population of susceptible individuals during the mean infectious time $T$ (see~\cite{Hethcote1976,VandenDriessche2017,Hethcote2000} and references therein). It is easy to see that the value of $\mathcal R_0$ for the modified SIR model \eqref{eq:modifiedSIR} is exactly the same as the $\mathcal R_0$ for the SIR model \eqref{eq:SIR}, \emph{i.e.}  $\mathcal R_0 = \beta / \alpha$ for both models. Note that this is in full agreement with \eqref{eq:peak_cond} in the sense that when $\mathcal R_0 > 1$ the infection survives, but when $\mathcal R_0 < 1$ the infection spontaneously disappears (see Section 4 for a more careful analysis of the fixed-point structure of both models). 

The identification between basic reproduction numbers for both models is a direct consequence of the fact that, for initial conditions $S_0 \approx 1$ and $I_0 \approx 0$, the early dynamics of systems \eqref{eq:SIR} and \eqref{eq:modifiedSIR} are similar. More in detail, it is well-known (see for instance \cite{VandenDriessche2017}) that the initial dynamics of the SIR model under such initial conditions is given by 
\begin{equation}
I(t) \approx I_0 e^{\alpha (\mathcal R_0 - 1) t} ,
\end{equation}
while for the modified SIR the closed-form solution \eqref{eq:solsimply} shows that 
\begin{equation}
I(t) = \frac{I_0 \, e^{t/\tau}}{\left( S_0 + I_0 e^{t/\tau} \right)^{\beta \tau}} \approx I_0 \, e^{t/\tau} = I_0 e^{\alpha (\mathcal R_0 - 1) t},
\end{equation}
where we have used that $I_0 \ll S_0$, $S_0 \approx 1$ and $t \ll \tau$.

Moreover, the closed-form solution \eqref{eq:solsimply} allows the computation of a generalization of the basic reproduction number, the so-called replacement number $\mathcal R (t)$ \cite{Hethcote2000}. The function $\mathcal R (t)$ is defined as the average number of secondary cases produced by one infected individual during the mean infectious time $T$, where the infected individual is introduced in a population that is in an arbitrary state of the infection outbreak. In our case, since the rate $r$ of secondary infections is given by the term ${\beta\,S\,I}/({S+I})$ in~\eqref{eq:modifiedSIR}, taking into account that $T=1/\alpha$, we obtain
\begin{equation}
\label{eq:rn}
{\cal{R}}(t)=\frac{r\, T}{I}=\frac{\beta}{\alpha} \left( \dfrac{S}{S+I} \right)=\frac{\beta}{\alpha} \,\left( \frac{1}{1+ \frac{I_0}{S_0} e^{t/\tau}} \right)
\, .
\end{equation}
It is clear that $\mathcal R_0 > \mathcal R (t)$ for all $t \in \mathbb R$. Moreover, $\mathcal R_0$ is given by~\eqref{eq:rn} when $t \ll \tau$ and $I_0 \ll S_0$. In fact, we could also say that
\begin{equation}
\mathcal R_0 = \lim_{(S,I) \to (1,0)} \mathcal R (t) = \lim_{t \to -\infty} \mathcal R (t) = \frac{\beta}{\alpha} \, ,
\end{equation}
where this expression should be thought of as a way of reversing the dynamics towards the point $S \approx 1$ and $I \approx 0$. Graphically, this means that we are moving along the solution depicted in Figure \ref{fig:pd_modSIR} in the flow opposite direction, in order to arrive to the $S$-axis. Note that this limit is independent of the arbitrary time origin used to define the initial conditions for the system of ODEs. 

Another interesting insight is gained by computing the area below the infective curve $I(t)$. In order to do that, we do not even need to perform the integration of $I(t)$, since from the third equation in \eqref{eq:modifiedSIR} we get 
\begin{equation}
Area(I) = \int^\infty_0 I(t) dt = \frac{1}{\alpha} \int^\infty_0 \dot R(t) dt = \lim_{t \to \infty} R(t) - R(0) = \frac{S_0 + I_0}{\alpha} = \frac{1}{\alpha}.
\end{equation}
This result is specially interesting from a parameter estimation point of view, since it allows to obtain a value for $\alpha$ directly from the data. Afterwards, assuming that $S_0$ and $I_0$ are known, $\beta$ can be obtained, for instance, from \eqref{eq:tpeak}. Thus, the exact solution in closed form greatly simplifies the fitting with actual data. Moreover, as we will see below, since the dynamics of the SIR and modified SIR systems are quite close (for a realistic range of the parameters), this procedure for the determination of the parameters of the modified model provides a good approximation for the parameters of the SIR one. 

A related interesting quantity from the epidemiological point of view is the removal rate, defined by $\dot R(t)=\alpha I$. While for the SIR model it can only be approximated by a closed-form expression in certain limits (see \cite{Bailey1975}), in the modified SIR system it can obviously computed exactly. Therefore, the behaviour of the removal rate (divided by $\alpha$) for the modified SIR system can be directly extracted from Figure \ref{fig:comp_traj}.

For any epidemic outbreak, it is also enlightening to analyze the intersection of the susceptible $S(t)$, infective $I(t)$ and recovered $R(t)$ functions. The closed-form solution of the modified SIR model allows us to get some exact results in this respect, which we write down in the following 

\begin{proposition}
\label{eq:crossSIR}
For the modified SIR system given by \eqref{eq:modifiedSIR}, with $\beta > \alpha$ and initial conditions $S(0)=S_0$, $I(0)=I_0$, $R(0)=0$ such that $S_0 > I_0 > 0$ and $S_0 > \alpha/\beta$, any two of the curves $S(t)$, $I(t)$ and $R(t)$ always intersect exactly once, regardless of the exact values of the initial conditions and parameters of the system. 

Furthermore:
\begin{itemize}
\item[i)] The curves $S(t)$ and $I(t)$ intersect before the infection peak if $\beta > 2 \alpha$, exactly at the infection peak if $\beta = 2 \alpha$ and after the infection peak if  $\beta < 2 \alpha$.

\item[ii)] The three curves $S(t)$, $I(t)$ and $R(t)$ intersect in a common point if and only if $\frac{\beta}{\alpha} < \frac{\log 3}{\log 3 - \log 2}$ and $S_0 = \frac{1}{3} \left( \frac{3}{2} \right)^{\beta / \alpha}$.

\item[iii)] The three curves $S(t)$, $I(t)$ and $R(t)$ intersect exactly at the infection peak if and only if $\beta = 2 \alpha$ and $S_0 = \frac{3}{4}$.

\end{itemize}
\end{proposition}

\begin{proof}

The solution of \eqref{eq:modifiedSIR} when $R(0)=0$ is given by \eqref{eq:solsimply}. In particular, the unique time at which the curves $S(t)$ and $I(t)$ intersect can be explicitly computed from \eqref{eq:ISexp}, and it reads 
\begin{equation}
\label{eq:tSI}
t_{SI} = \tau \log \left( \frac{S_0}{I_0} \right) .
\end{equation}
Since we are assuming $S_0 > I_0$, this time is always positive.  

From \eqref{eq:solsimply} we can also compute the times $t_{SR}$ and $t_{IR}$ such that $R(t_{SR}) = S(t_{SR})$ and $R(t_{IR}) = I(t_{IR})$. It is easy to check that these times are given by the common expression
\begin{equation}
\label{eq:tSRIRX}
t^* = \tau \log \left( \frac{X-S_0}{I_0} \right) ,
\end{equation}
where $X$ is a solution of the equation
\begin{equation}
\label{eq:XtSR}
X^{\beta \tau} - X - S_0 = 0 
\end{equation}
in the case of $t_{SR}$, while $X$ is a solution  of the equation
\begin{equation}
\label{eq:XtIR}
X^{\beta \tau} - 2X + S_0 = 0 
\end{equation}
in the case of $t_{IR}$. An elementary computation shows that equation \eqref{eq:XtSR} has only one solution, and therefore $t_{SR}$ is unique. Equation \eqref{eq:XtIR} has 2 solutions, the first one living in $(0,1)$ and the second one within $(1,\infty)$. The first of these solutions results in a negative or complex time, so $t_{IR}$ is defined by the unique solution of \eqref{eq:XtIR} in $(1,\infty)$. Therefore, we have proved that $t_{SI}$, $t_{SR}$ and $t_{IR}$ are unique, which means that the curves $S(t)$, $I(t)$ and $R(t)$ intersect exactly once.

Now, statement $i)$ derives straightforwardly from a comparison between \eqref{eq:tSI} and \eqref{eq:tpeak}.

To prove $ii)$ we first note that $t^* = t_{SI}$ if and only if $X=2 S_0$, which is a solution of both \eqref{eq:XtSR} and \eqref{eq:XtIR} if and only if
\begin{equation}
\label{eq:S0triple}
S_0 = \frac{1}{3} \left( \frac{3}{2}\right)^{\beta/\alpha} .
\end{equation}
Given that $S_0 < 1$, then $\frac{1}{3} \left( \frac{3}{2}\right)^{\beta/\alpha} < 1$, and therefore $\frac{\beta}{\alpha} < \frac{\log 3}{\log (3/2)}$.

Statement $iii)$ is a consequence of $i)$ and $ii)$, since in order that the intersection coincides with the infection peak we need that $\beta = 2 \alpha$, and substituting this into the condition for triple intersection \eqref{eq:S0triple} we get $S_0 = \frac{3}{4}$. Equivalently, note that if $\beta = 2 \alpha$, then by \eqref{eq:Ipeak}, $S(t_{\text{peak}}) = I(t_{\text{peak}}) = \frac{1}{4 S_0}$, so $R(t_{\text{peak}}) = 1 - \frac{1}{2 S_0}$, and by imposing that they coincide we get $S_0 = \frac{3}{4}$.

\vspace{-0.5cm}
\end{proof}

\vspace{-0.7cm}
\paragraph{Remark 3}
The condition $S_0 > \alpha / \beta$ in the previous Proposition is just the condition for the existence of an infectious peak \eqref{eq:peak_cond2} but we are not using it explicitly in the proof. Therefore, all the previous results which do not involve the infection peak are true if this condition is removed. 
\vspace{-0.5cm}
\paragraph{Remark 4}
The fact that the functions $S(t)$, $I(t)$ and $R(t)$ always intersect is a striking difference with respect to the original SIR system of Kermack and McKendrick \eqref{eq:SIR}.

\section{Fixed points and comparison with the SIR system}

Finally, it is worth performing a more detailed analysis of the dynamics of the modified SIR system \eqref{eq:modifiedSIR} when compared to the original SIR model \eqref{eq:SIR}. Although the exact closed-form solution here obtained in the modified case is valid for any values of $\alpha$ and $\beta$, for the sake of brevity from now on we only consider the case $\beta > \alpha$ (recall from \eqref{eq:peak_cond} that this is the regime in which an actual outbreak does exist), or equivalently, $\mathcal R_0 > 1$.

Qualitatively, the most significant difference between both dynamical systems is the stability of their fixed points. From equations \eqref{eq:SIR} and \eqref{eq:modifiedSIR} we see that $I=0$ is a line of non-isolated fixed points for both models. The stability of any of these fixed points $p=(S,0)$ is defined simply by the trace of the Jacobian evaluated at this point, $\mathbf{J} (p)$. For the SIR system this trace is 
\begin{equation}
\label{eq:trSIR}
\mathrm{Tr}\; \mathbf{J} (p) = \beta S - \alpha ,
\end{equation}
while for the modified SIR system it reads
\begin{equation}
\label{eq:trmodSIR}
\mathrm{Tr}\; \mathbf{J} (p) = \beta - \alpha .
\end{equation}
Therefore, for the modified SIR system all the points belonging to the line $I=0$ (with the exception of $(S,I)=(0,0)$) are unstable (recall that we only consider the epidemiologically relevant case $\beta > \alpha$). Note that this agrees with the value for the basic reproduction number given in Section 3. Meanwhile, for the SIR system points such that $S>\alpha/\beta$ are unstable, while points with $S<\alpha/\beta$ are stable. This can be clearly appreciated in Figures \ref{fig:pd_SIR} and \ref{fig:pd_modSIR}, where the corresponding flows are presented for both systems. Colored curves correspond to the phase space equation $I(S)$ for each model. Trajectories of the system starting at any point of the appropriate curve will follow this curve (in the direction of the flow) in order to reach the relevant fixed point.

The differences regarding the fixed point structure of these two systems can also be analyzed algebraically. For the SIR system \eqref{eq:SIR}, it is well-known that the phase space equation is 
\begin{equation}
I(S) = \frac{\alpha}{\beta} \log S - S + \mathcal C,
\end{equation}
where $\mathcal C$ is a constant (this is the equation of the curves in Figure \ref{fig:pd_SIR}, for different values of $\mathcal C$). In fact, 
\begin{equation}
\label{eq:C_SIR}
\mathcal C = S + I - \frac{\alpha}{\beta} \log S
\end{equation}
is the Casimir function for the associated Poisson structure (see \cite{BBG2020hamiltonianepidemics} for details). It is easy to prove that the equation $I(S) = 0$ always have a solution $S \in (0,\alpha/\beta)$. However, the phase space equation $I(S) = 0$ for the modified SIR system, where $I(S)$ is given by \eqref{eq:I_S}, always has $S=0$ as a solution (see Figure \ref{fig:pd_modSIR}). This equation is directly obtained from the Casimir function \eqref{eq:C}, and it is interesting to compare this Casimir function \eqref{eq:C} with the exponential of \eqref{eq:C_SIR}.

\medskip

\paragraph{Remark 4}
The previous discussion shows that the different qualitative behaviour of the systems \eqref{eq:SIR} and \eqref{eq:modifiedSIR} can be algebraically understood through the differences between the Casimir functions \eqref{eq:C} and \eqref{eq:C_SIR}, and in particular, the different structure and location of their zeroes within the phase space. 

From the epidemiological point of view, the existence of stable fixed points (different from $(0,0)$) in the $I=0$ axis  explains the well-known fact that in the original SIR model of Kermack and McKendrick the whole population is not infected during the evolution of the infection. While these results can be obtained from a dynamical systems approach, it is interesting to note their direct connection with the algebraic and geometric structure of the Poisson manifold underlying their description of epidemiological models as generalized Hamiltonian systems. 

\vspace{-0.5cm}

\paragraph{Remark 5} For the modified SIR system, the closed-form analytical solution \eqref{eq:sol} contains all the previous information. Solutions with $I_0=0$ are constant functions, and 
\begin{equation}
\lim_{t \to \infty} S(t) = 0, \qquad \lim_{t \to \infty} I(t) = 0, \qquad \lim_{t \to \infty} R(t) = 1 , 
\end{equation}
for any initial conditions such that $I_0 \neq 0$. This shows that the only stable fixed point is $(S,I) = (0,0)$. Note that for the modified SIR system it takes an infinite time to reach this fixed point, regardless of the initial conditions.

In Figure \ref{fig:comp_traj}, plots for some trajectories $S(t)$ and $I(t)$ contained in the phase space orbits from Figures \ref{fig:pd_SIR} and \ref{fig:pd_modSIR} are depicted. In the left column $\beta = 1$ and $\alpha = 0.2$ while in the right column $\beta = 1$ and $\alpha = 0.6$. Each plot contains four different curves: coloured ones correspond to $I(t)$ while black ones correspond to $S(t)$, and solid ones correspond to the modified SIR system while dashed ones correspond to the original SIR system. The first row shows the dynamics for initial conditions such that the outbreak rapidly extinguishes. The second row shows the limiting case given by $S_0 = \alpha / \beta$ (note that this value is the same for both models since $S_0+I_0 = 1$). The third row shows the typical behavior for values of the parameters and initial conditions for which there is an actual outbreak, and therefore $I(t)$ has a maximum. The fourth row shows a situation such that at the beginning of the outbreak a fraction of the total population is immunized.

\begin{figure}[H]
\begin{center}
\includegraphics[width=.4\textwidth, height=.3\textwidth]{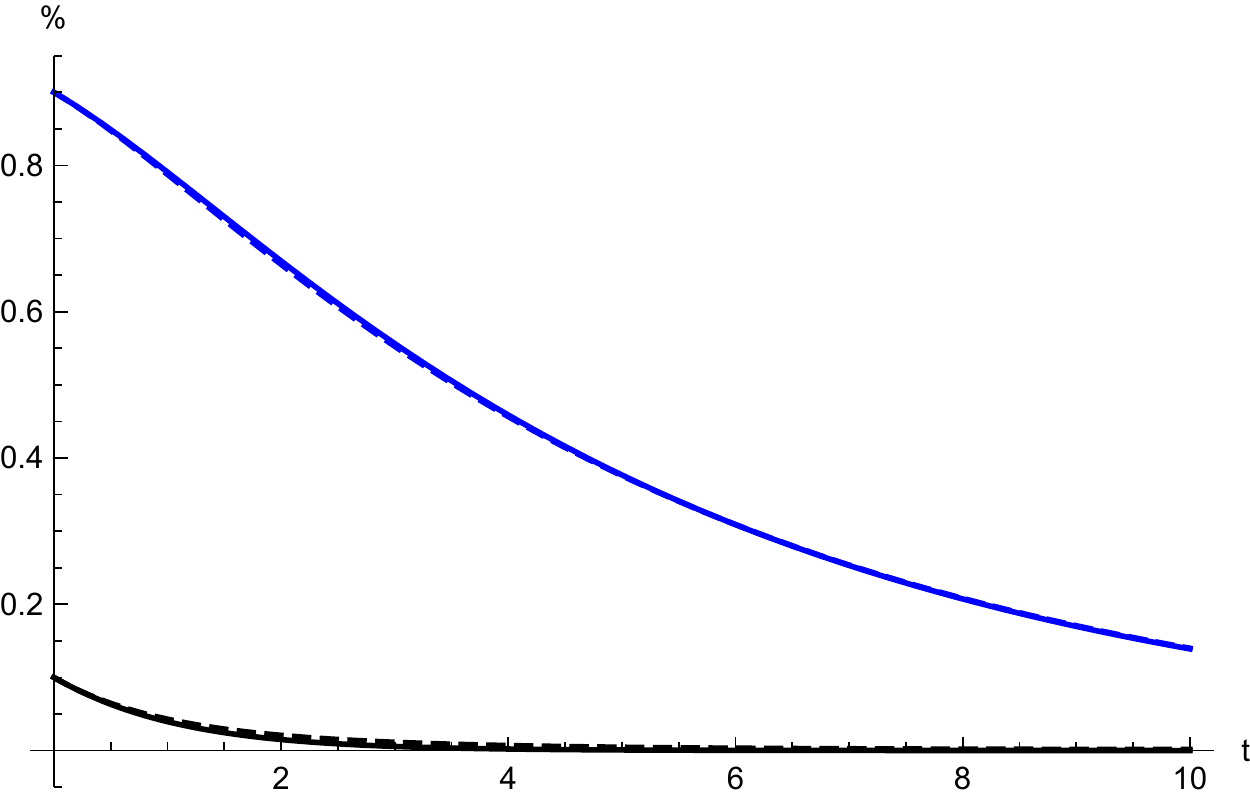} \includegraphics[width=.4\textwidth, height=.3\textwidth]{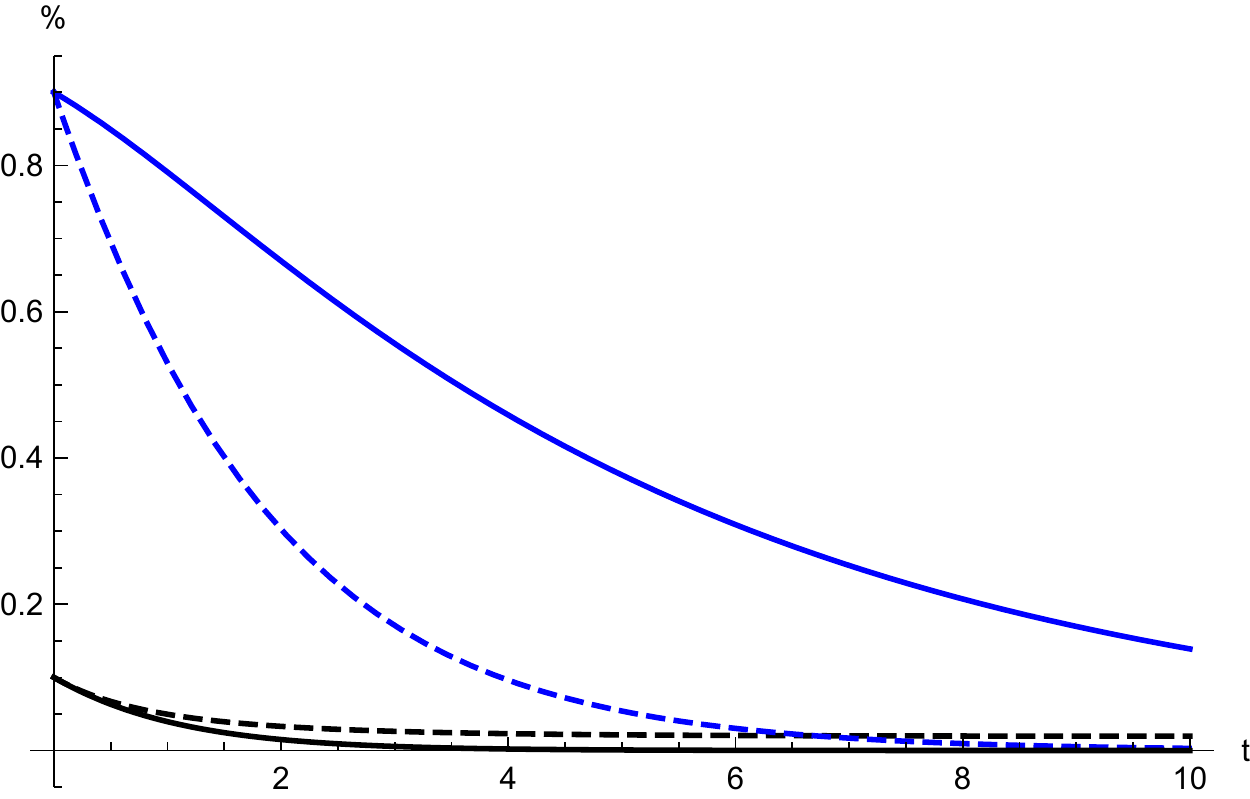} \\
\includegraphics[width=.4\textwidth, height=.3\textwidth]{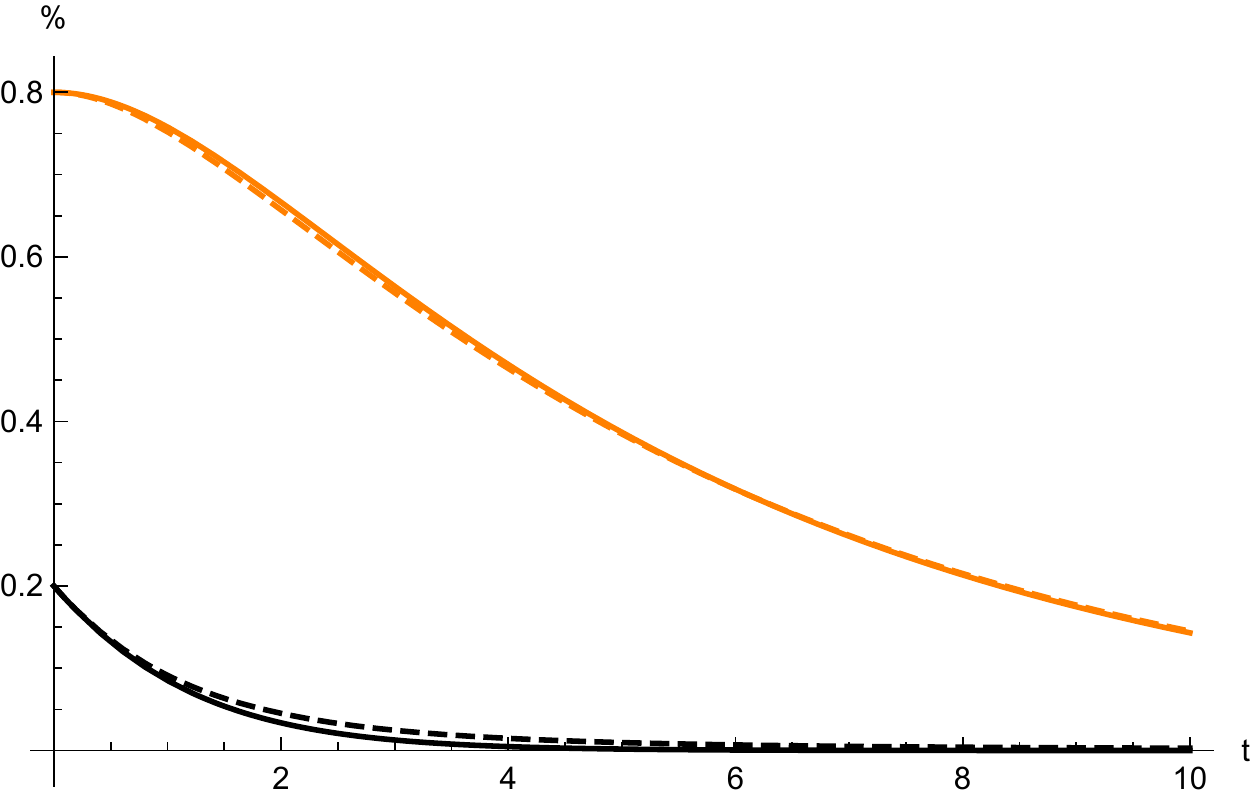} \includegraphics[width=.4\textwidth, height=.3\textwidth]{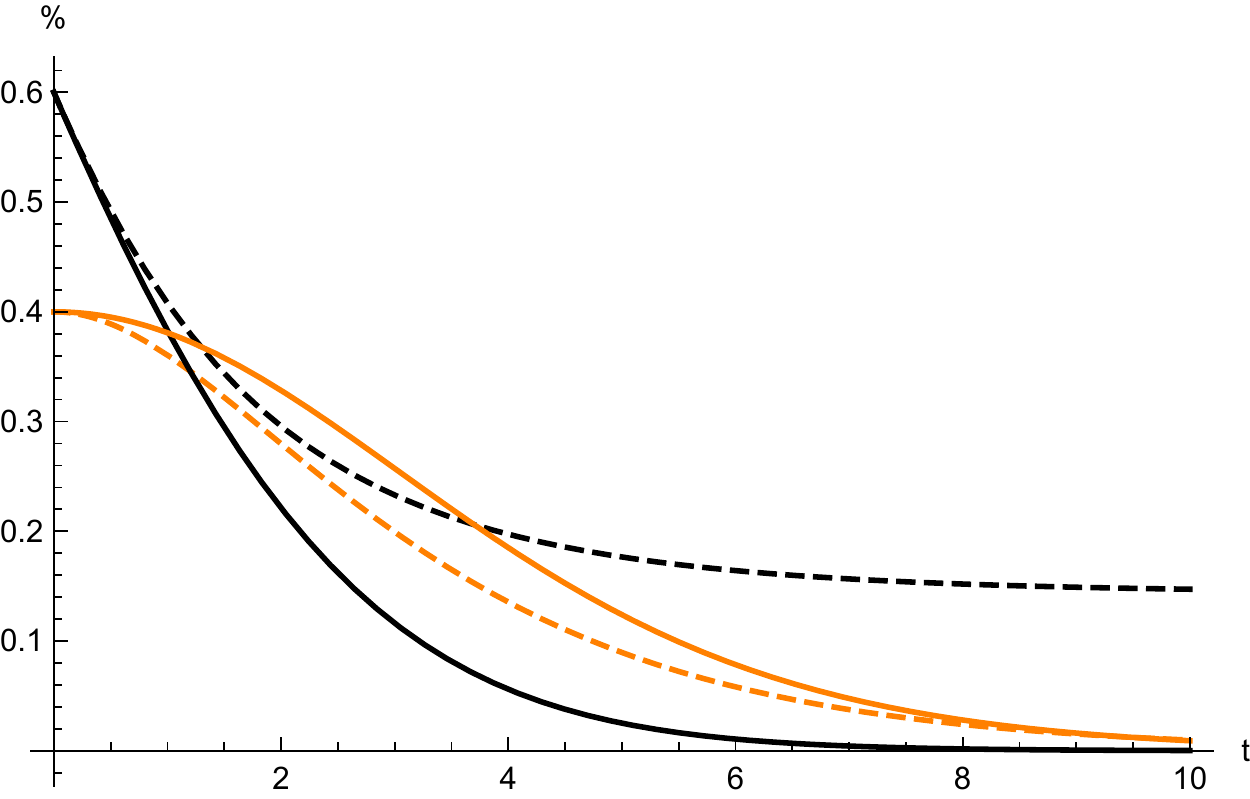} \\
\includegraphics[width=.4\textwidth, height=.3\textwidth]{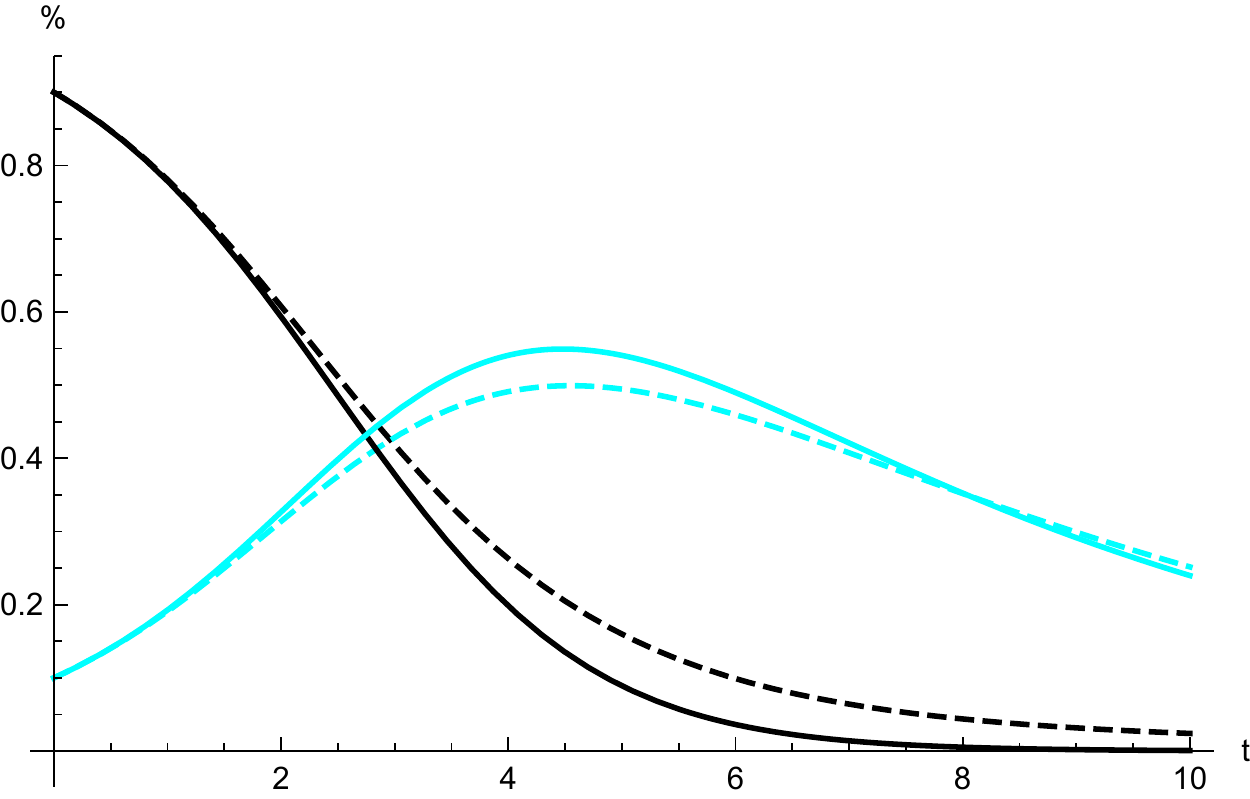} \includegraphics[width=.4\textwidth, height=.3\textwidth]{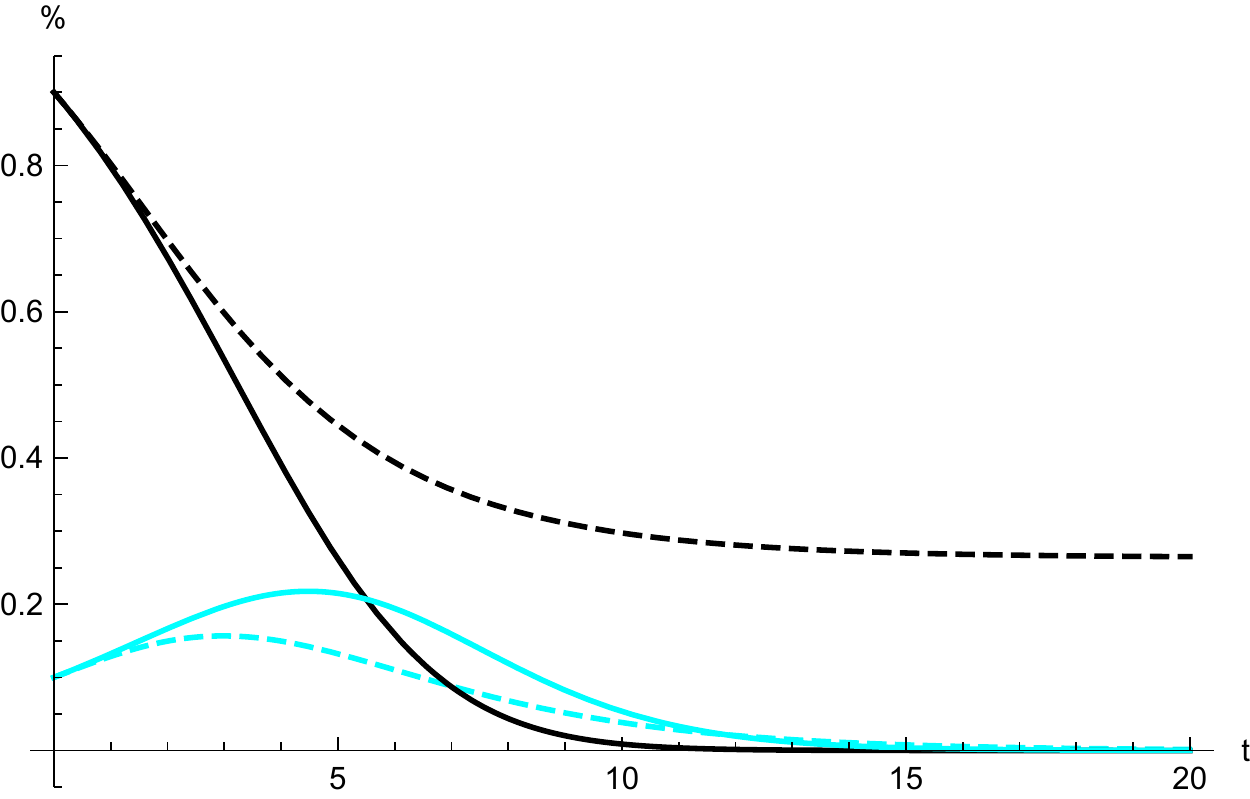} \\
\includegraphics[width=.4\textwidth, height=.3\textwidth]{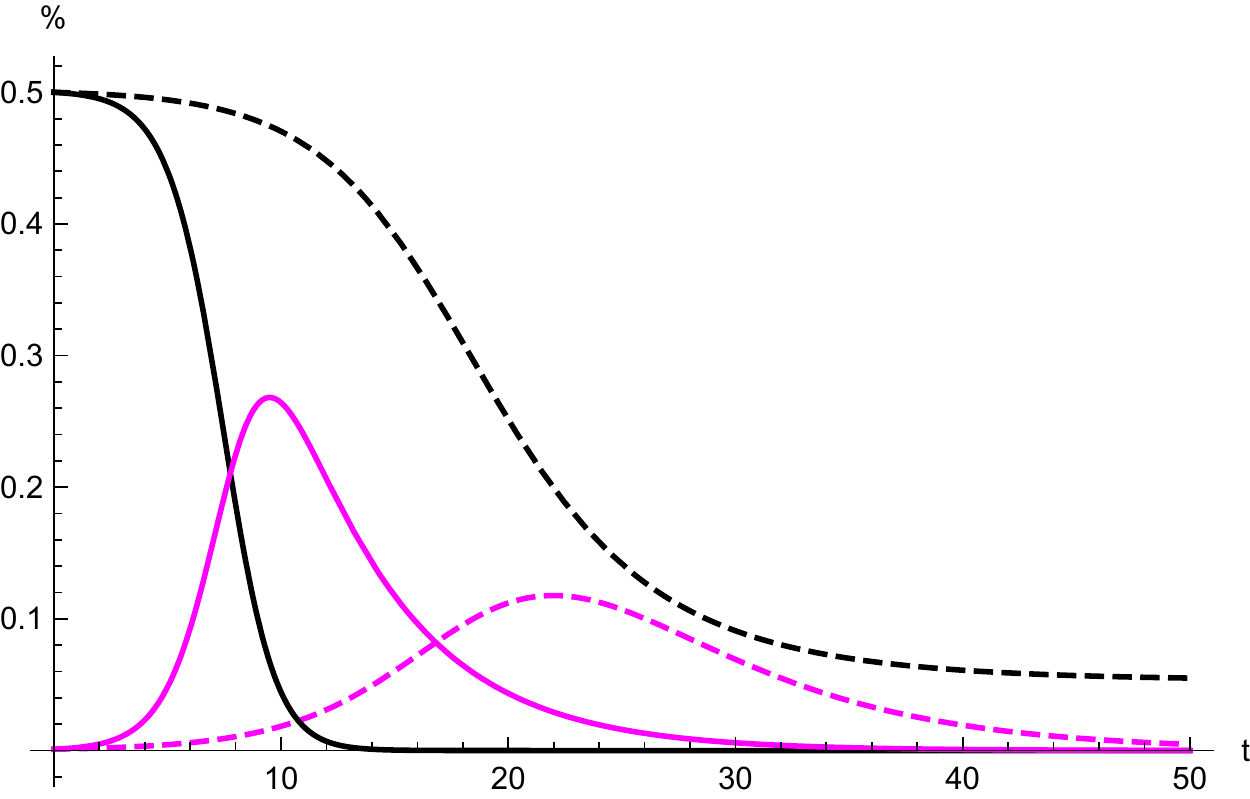} \includegraphics[width=.4\textwidth, height=.3\textwidth]{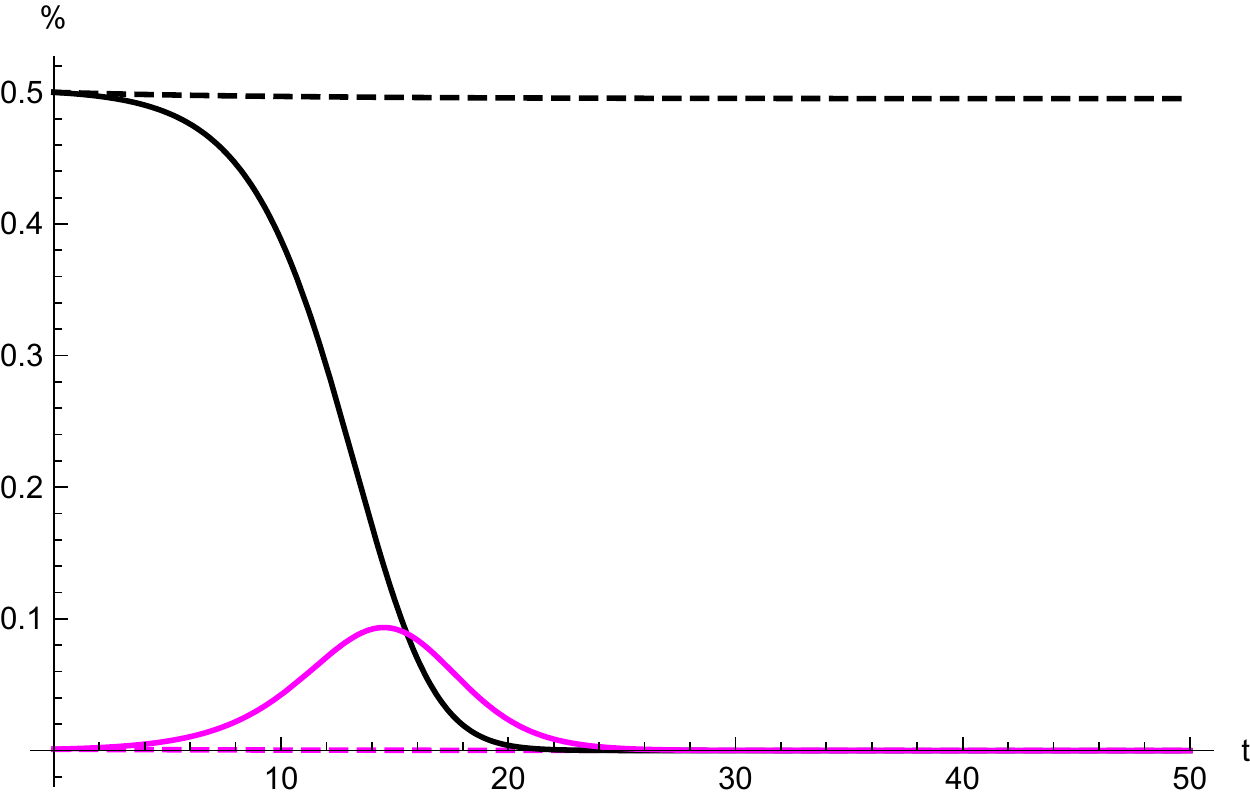} \\
\end{center}
\vspace{-0.5cm}
\caption{\small{ $S(t)$ (black) and $I(t)$ (colored) functions for the SIR system (dashed) and the modified SIR system (solid). $\beta = 1$. Left: $\alpha=0.2$. Right: $\alpha=0.6$.
Blue: $S_0 = 0.1$, $I_0 = 0.9$. Orange: $S_0 = \alpha / \beta$, $I_0 = 1- \alpha / \beta$. Cyan: $S_0 = 0.9$, $I_0 = 0.1$. Magenta: $S_0 = 0.5$, $I_0 = 0.001$.
\label{fig:comp_traj}
}}
\end{figure}

\hspace{.02\textwidth}

\begin{figure}[H]
\begin{center}
\includegraphics[width=.4\textwidth, height=.4\textwidth]{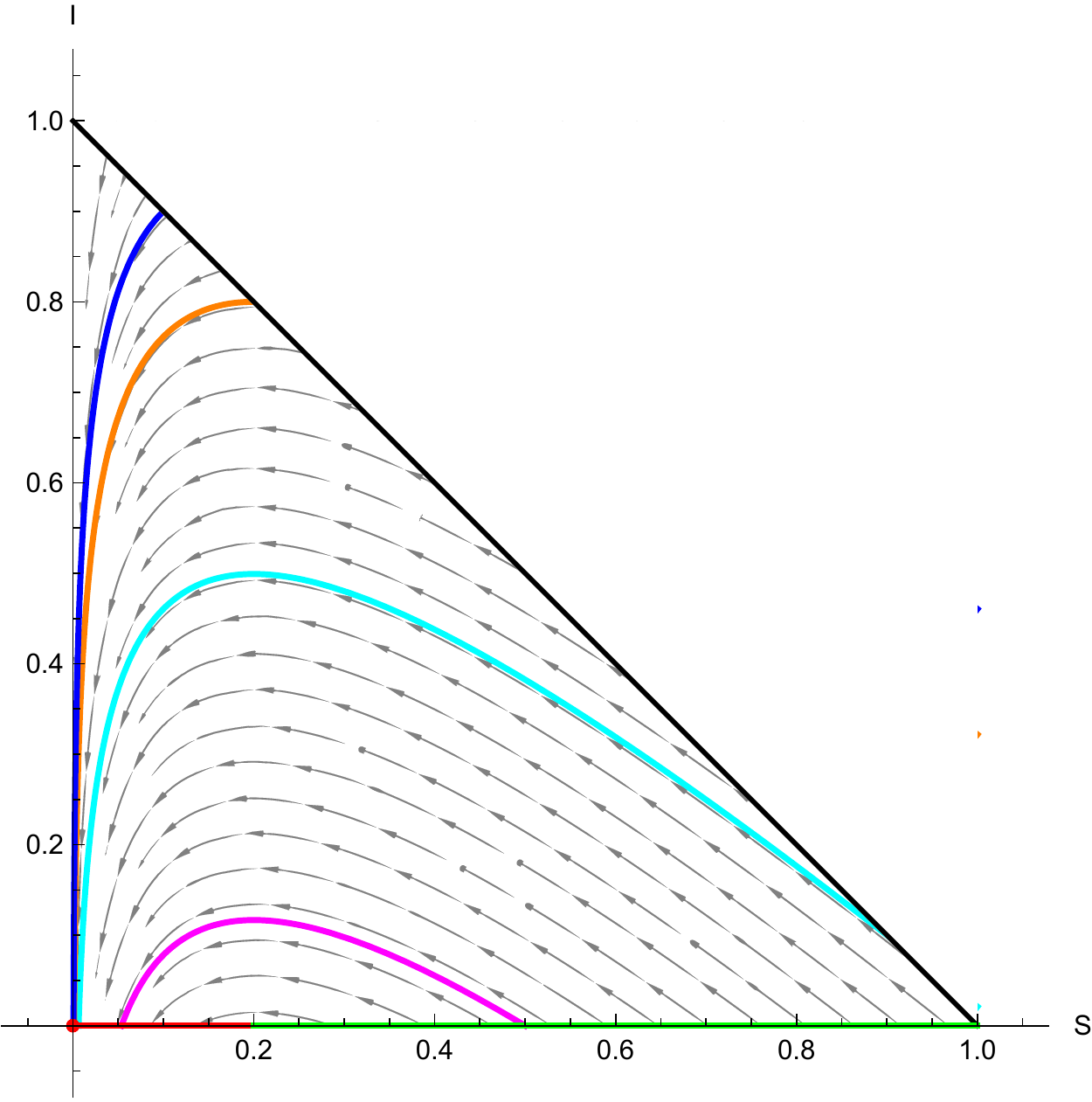}
\includegraphics[width=.4\textwidth, height=.4\textwidth]{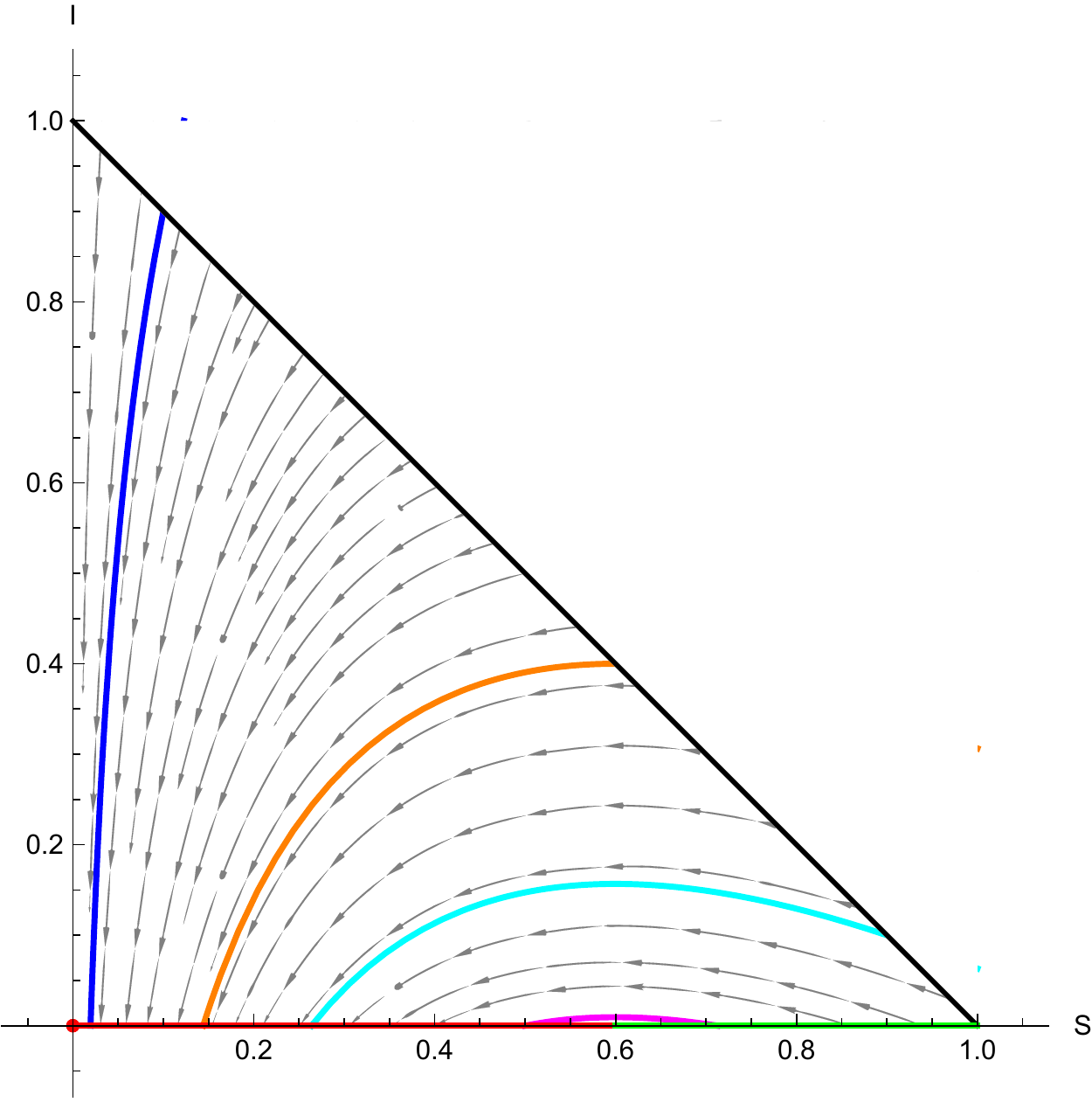}
\end{center}
\vspace{-0.5cm}
\caption{\small{ Phase space for the SIR system \eqref{eq:SIR}. $\beta = 1$. Left: $\alpha=0.2$. Right: $\alpha=0.6$. 
Blue line: $S_0 = 0.1$, $I_0 = 0.9$. Orange line: $S_0 = \alpha / \beta$, $I_0 = 1- \alpha / \beta$. Cyan line: $S_0 = 0.9$, $I_0 = 0.1$. Magenta line: $S_0 = 0.5$, $I_0 = 0.$ Red: Stable points. Green: Unstable points.
\label{fig:pd_SIR}
}}
\end{figure}
\begin{figure}[H]
\begin{center}
\includegraphics[width=.4\textwidth, height=.4\textwidth]{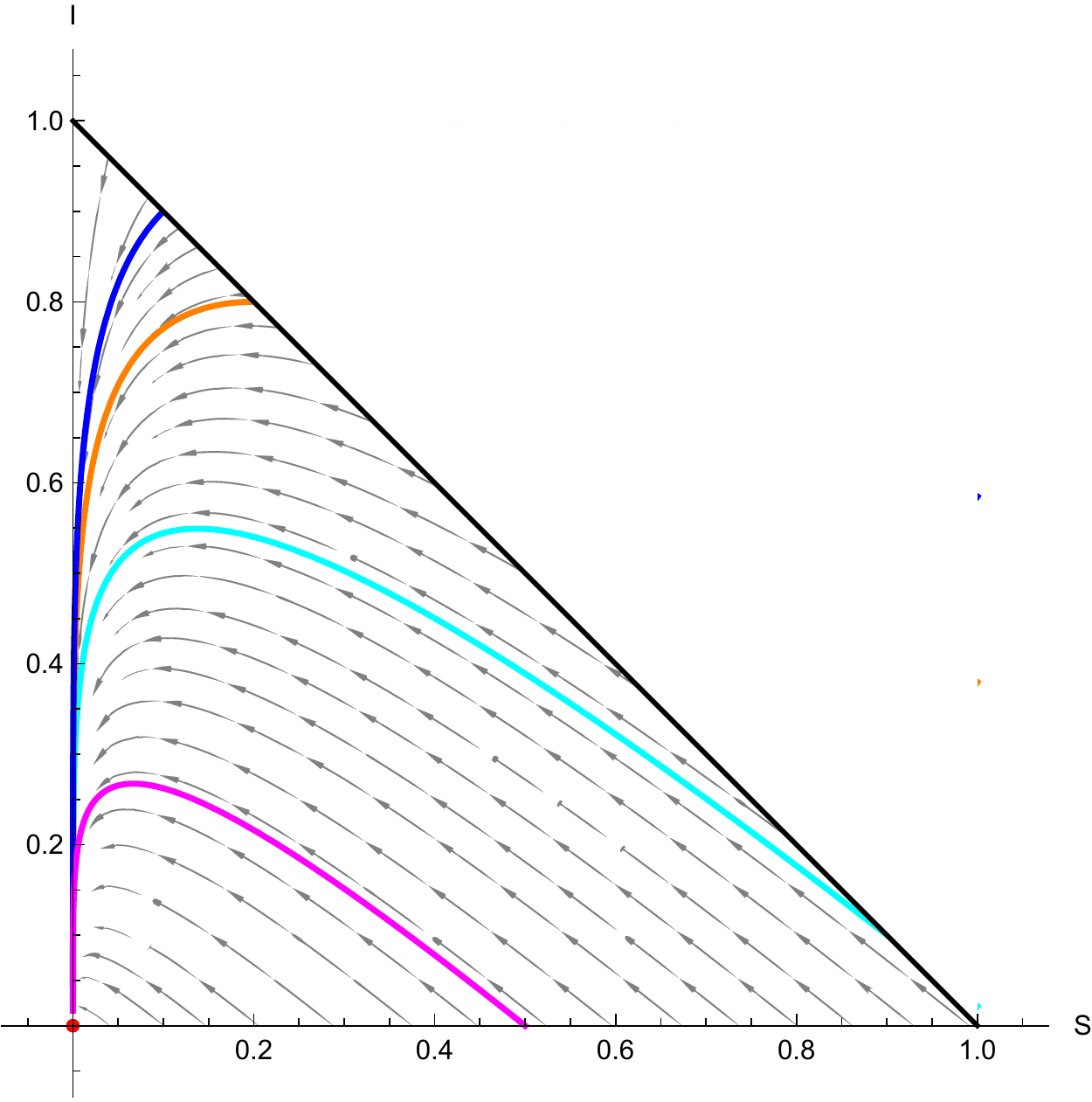}
\includegraphics[width=.4\textwidth, height=.4\textwidth]{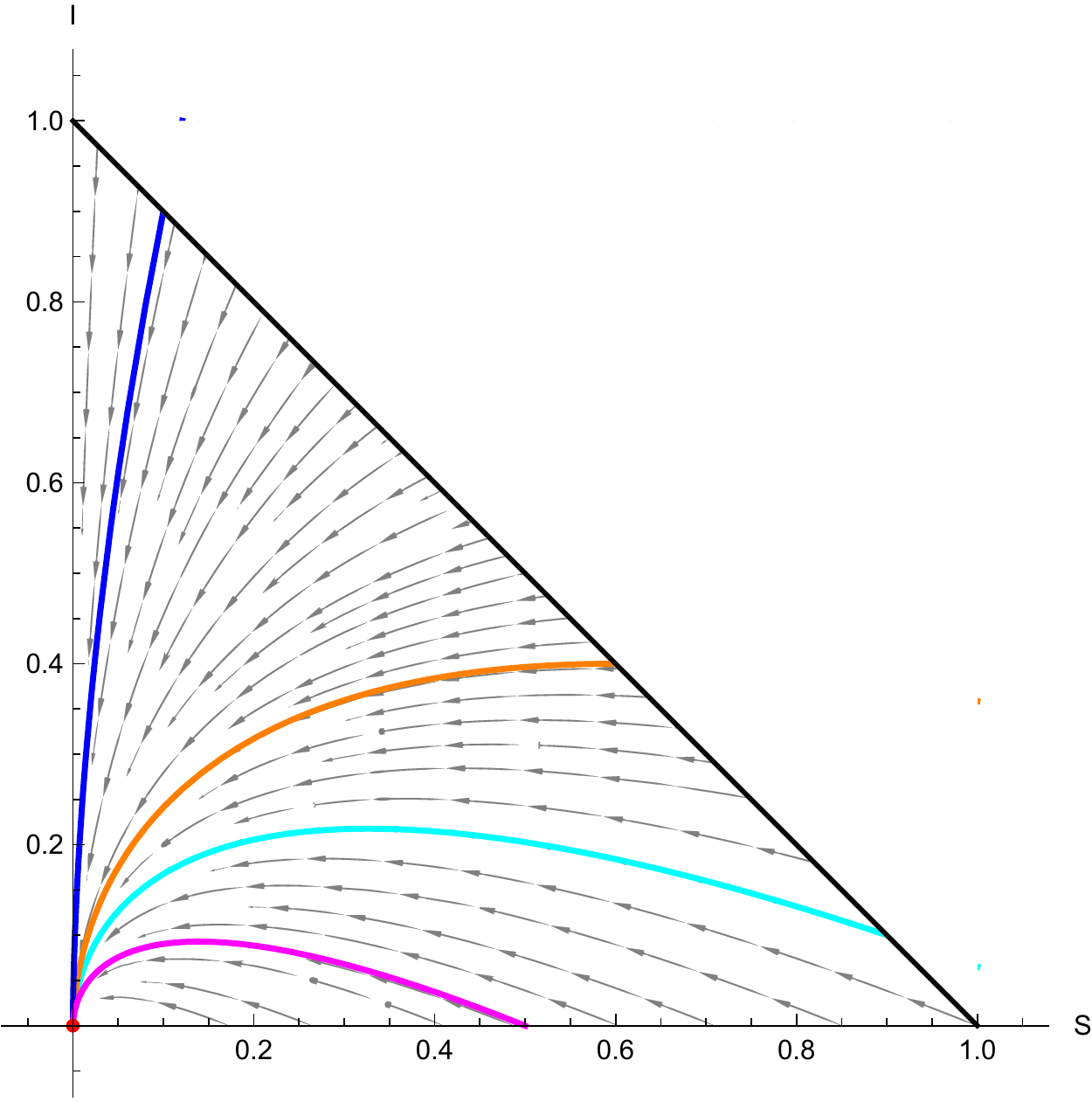}
\end{center}
\vspace{-0.5cm}
\caption{\small{ Phase space for the modified SIR system \eqref{eq:modifiedSIR}. $\beta = 1$. Left: $\alpha=0.2$. Right: $\alpha=0.6$.
Blue line: $S_0 = 0.1$, $I_0 = 0.9$. Orange line: $S_0 = \alpha / \beta$, $I_0 = 1- \alpha / \beta$. Cyan line: $S_0 = 0.9$, $I_0 = 0.1$. Magenta line: $S_0 = 0.5$, $I_0 = 0.$
\label{fig:pd_modSIR}
}}
\end{figure}

These plots show  all possible qualitatively different dynamics for the SIR and modified SIR systems. Essentially, as far as the ratio $\alpha/\beta$ grows, stronger differences between both models arise. In the left column we have $\alpha/\beta = 1/5$ and the dynamics of both systems are quite close (for the 3 first rows). In the right column $\alpha/\beta = 3/5$ and stronger differences appear, specially for $S(t)$. The most striking difference between both systems can be appreciated in the picture located at the last row, second column, which corresponds to a small perturbation of the case when initially half of the population is immunized and the other half is susceptible to the infection. In this case, the SIR system predicts no outbreak (it is a stable fixed point), while the modified SIR system does predict it. Obviously, this is due to the fact that in the modified SIR system we are assuming that the recovered population has been removed (death, quarantine, etc) and therefore does not interact (thus not contributing to the so-called `herd immunity', which is of course not attainable in this model). All these considerations can also be deduced from the phase space representation in Figures \ref{fig:pd_SIR} and \ref{fig:pd_modSIR}. 

However, it is important to stress that, very often, the  most relevant scenario from the epidemiological viewpoint is the one in which there are no immunized individuals at the beginning (or they are very few ones), which is given by the third row (cyan). We can conclude that in this scenario, specially when the ratio $\alpha/\beta$ is smaller, the SIR and modified SIR systems present similar features, with the modified SIR model always predicting a larger infection peak that the SIR.


\section{Modified SIR system with a time-dependent transmission rate}

We will now consider the modified SIR system where the transmission rate $\beta(t)$ becomes time-dependent, while we keep the recovery rate $\alpha$ as a constant value. As we mentioned in the introduction, the variability of the transmission rate along the evolution of a given disease can be usually related to modifications in the behavior of either the host or the pathogen/vector, which can be in some cases due to seasonal effects. Therefore, it seems natural to wonder whether such  time-dependent modified SIR models with smoothly varying $\beta(t)$ do admit also exact solutions in closed form (models with abrupt changes in time of the transmission rate are also considered in the literature~\cite{Liu}).

This question can be answered in the affirmative by recalling Remark 2 in Section 2, where the new variable $y(t) = {I(t)}/{S(t)} $ was considered. It is straightforward to check that the modified SIR model with $\beta(t)$ leads to a formally equivalent  system~\eqref{eq:MSIRy} (and the same would happen when a time-dependent recovery rate $\alpha(t)$ is simultaneously considered), namely
\begin{equation}
\label{eq:MSIRydept}
\dot{S} = -\frac{\beta(t)}{1+y}\,S\,y, \qquad\qquad\qquad \dot{y} = (\beta(t)-\alpha) \,y, \qquad\qquad\qquad \dot R = \alpha\,S\,y , \\
\end{equation}
for any time-dependent transmission rate function $\beta(t)$. Therefore, the  equation for $y(t)$ can be always integrated
\begin{equation}
\label{eq:MSIRyt}
y(t) = \frac{I(t)}{S(t)} = \frac{I_0}{S_0}e^{\int_0^t \left( \beta (s) - \alpha \right) \mathrm d s} , \,
\end{equation}
and its solution will be given in closed form provided  that the function $\beta(t)$ admits a well-defined primitive function. As a consequence,
we can write the explicit solution for $S(t)$ in the form
\begin{equation}\label{eq:SinMSIRydept}
S(t) = S_0 e^{-\int_0^t \frac{\beta(s) y(s)}{1+y(s)} \mathrm d s} = S_0 \exp \left( - \int_0^t \beta(s)  \left( \frac{\frac{I_0}{S_0} e^{\int_0^s \left( \beta (u) - \alpha \right) \mathrm d u}}{1 + \frac{I_0}{S_0} e^{\int_0^s \left( \beta (u) - \alpha \right) \mathrm d u}} \right) \mathrm d s \right) ,
\end{equation}
and finally we have
\begin{equation}
I(t) = S(t) \,y(t) , \qquad \qquad R(t) = \alpha\, I(t) .
\end{equation}

As we will show in the sequel, there do exist (a priori realistic) functions $\beta(t)$ such that~\eqref{eq:SinMSIRydept} can be also given in closed form. In particular, we will work out  two models: the first one with a monotonically decreasing transmission rate, and the second one corresponding to a disease in which the transmission rate has a maximum at $t\neq 0$.  Finally, for the sake of completeness, an oscillating transmission rate representing periodic seasonal effects will be also considered, although in this case~\eqref{eq:SinMSIRydept} cannot be given in closed form.

\subsection{A model with a monotonically decreasing transmission rate}

Let us consider the transmission rate function
\begin{equation}
\beta (t) = \alpha + \frac{1/\tau}{1+t/\tau}, \qquad\mbox{with}\qquad  \tau = 1/(\beta_m - \alpha) \, .
\end{equation}
This function presents a maximum with value $\beta_m$ at $t=0$, and then monotonically decreases towards $\alpha$ (see Figure \ref{fig:comp_beta}). This function could model a situation in which the population would be aware of the presence of the virus at $t = 0$ and they take progressive measures in order to prevent its spread, like social-distancing, mask-wearing, etc. In this case the closed-form solution of the modified SIR model is given by
\begin{equation}
y(t) = \frac{I_0}{S_0} \left( 1 + \frac{t}{\tau}\right) ,
\end{equation}
\begin{equation}
S(t) = S_0 \left( \frac{S_0 + I_0}{S_0 + I_0 (1+t/\tau)} \right)^{1-\frac{S_0}{I_0} \alpha \tau} e^{-\alpha t} ,
\end{equation}
and 
\begin{equation}
I(t) = I_0 \left( 1 + \frac{t}{\tau}\right) \left( \frac{S_0 + I_0}{S_0 + I_0 (1+t/\tau)} \right)^{1-\frac{S_0}{I_0} \alpha \tau} e^{-\alpha t} .
\end{equation}

As we can appreciate in Figure \ref{fig:comp_traj}, in this model the height of the peak of the outbreak is smaller than in the modified SIR model with  $\beta$ constant. Such peak takes place at the time
\begin{equation}
t_{peak} = \tau \left( \sqrt{\frac{S_0}{I_0 \alpha \tau}} - 1 \right) \, ,
\end{equation}
which  can be given in terms of the parameters of the model due to the closed form for $I(t)$.

\subsection{A model with a maximum transmission rate for $t\neq 0$}

Let us now consider the transmission rate function (see Figure \ref{fig:comp_beta})
\begin{equation}
\beta (t) = \alpha + \frac{2 t/\tau^2}{1+t^2/\tau^2}, \qquad
\mbox{with}
\qquad  \tau = 1/(\beta_m - \alpha).
\end{equation}
This function presents a maximum of the transmission rate with value $\beta_m$ at $t=\tau$, which would model a strong initial intensification of the transmission rate followed by its monotonic decreasing. This second example would provide a toy model for a situation such as an abrupt cancellation of a confinement, where the population incorrectly assumes that the disease is no longer present, therefore resuming its usual interactions very fast, and afterwards correcting this behavior progressively. The exact solution of the model in closed form is given by
\begin{equation}
y(t) = \frac{I_0}{S_0} \left( 1 + \frac{t^2}{\tau^2}\right) ,
\end{equation}
\begin{equation}
S(t) = S_0  \left( \frac{S_0 + I_0}{S_0 + I_0 (1+t^2/\tau^2)} \right) 
\exp \left[
\alpha\left(
\dfrac{S_0}{S_0+I_0}
\right)\left(\dfrac{\arctan \left(\sqrt{\frac{I_0}{S_0+I_0}}\frac{t}{\tau}\right)}{\sqrt{\frac{I_0}{S_0+I_0}}\frac{1}{\tau}}\right)
\right]
 e^{- \alpha t},
\end{equation}
and 
\begin{equation}
I(t) = I_0 \left( 1 + \frac{t^2}{\tau^2}\right) \left( \frac{S_0 + I_0}{S_0 + I_0 (1+t^2/\tau^2)} \right) 
\exp \left[
\alpha\left(
\dfrac{S_0}{S_0+I_0}
\right)\left(\dfrac{\arctan \left(\sqrt{\frac{I_0}{S_0+I_0}}\frac{t}{\tau}\right)}{\sqrt{\frac{I_0}{S_0+I_0}}\frac{1}{\tau}}\right)
\right]
e^{- \alpha t} .
\end{equation}

As it is shown in Figure \ref{fig:comp_traj}, for the second model the infection peak is higher than the one corresponding to the first model, although the maximum value for $\beta$ is in both cases the same. This shows that for this kind of models 
not only the maximum value of the transmission rate function is relevant, but also the time at which the peak of $\beta(t)$ appears.

\begin{figure}[H]
\begin{center}
\includegraphics[width=.6\textwidth, height=.3\textwidth]{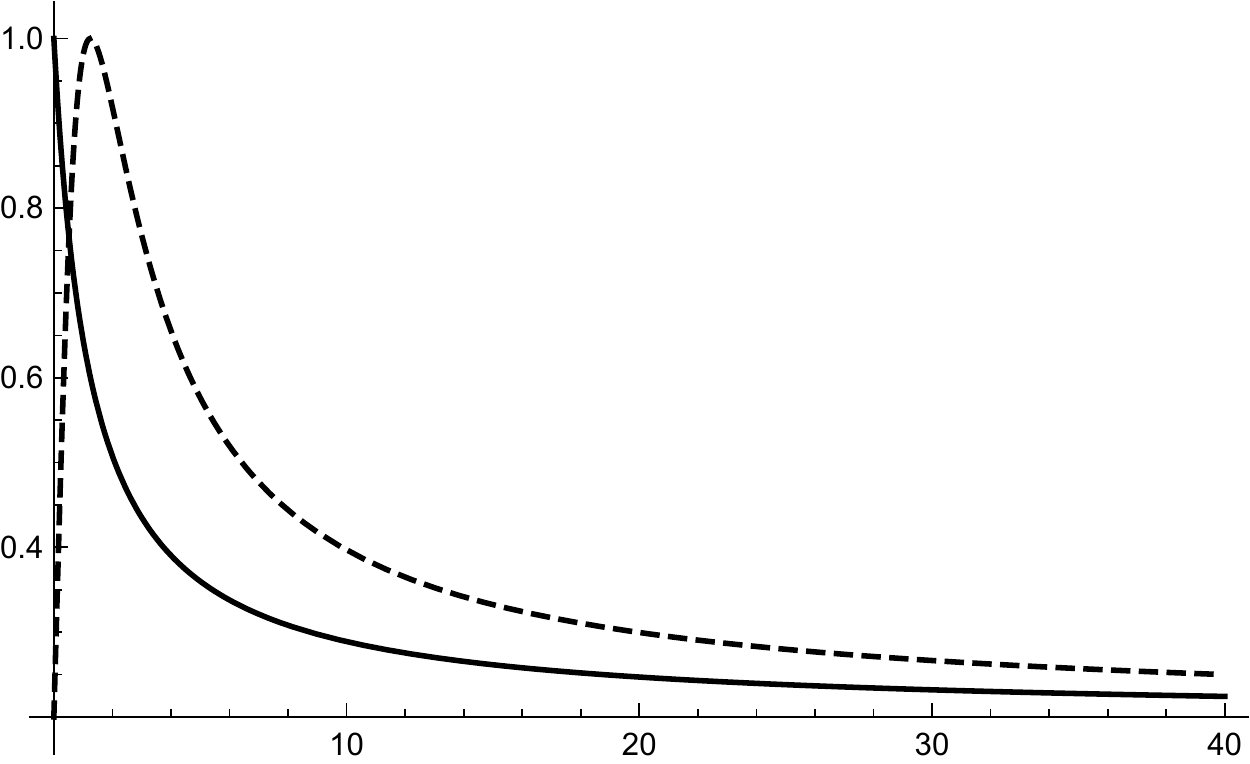} 
\end{center}
\vspace{-0.5cm}
\caption{\small{ Solid: $\beta (t) = \alpha + \frac{1/\tau}{1+t/\tau}$. Dashed: $\beta (t) = \alpha + \frac{2 t/\tau^2}{1+t^2/\tau^2}$. $\beta_m = 1, \alpha = 0.2$. 
\label{fig:comp_beta}
}}
\end{figure}

\begin{figure}[H]
\begin{center}
\includegraphics[width=.48\textwidth, height=.33\textwidth]{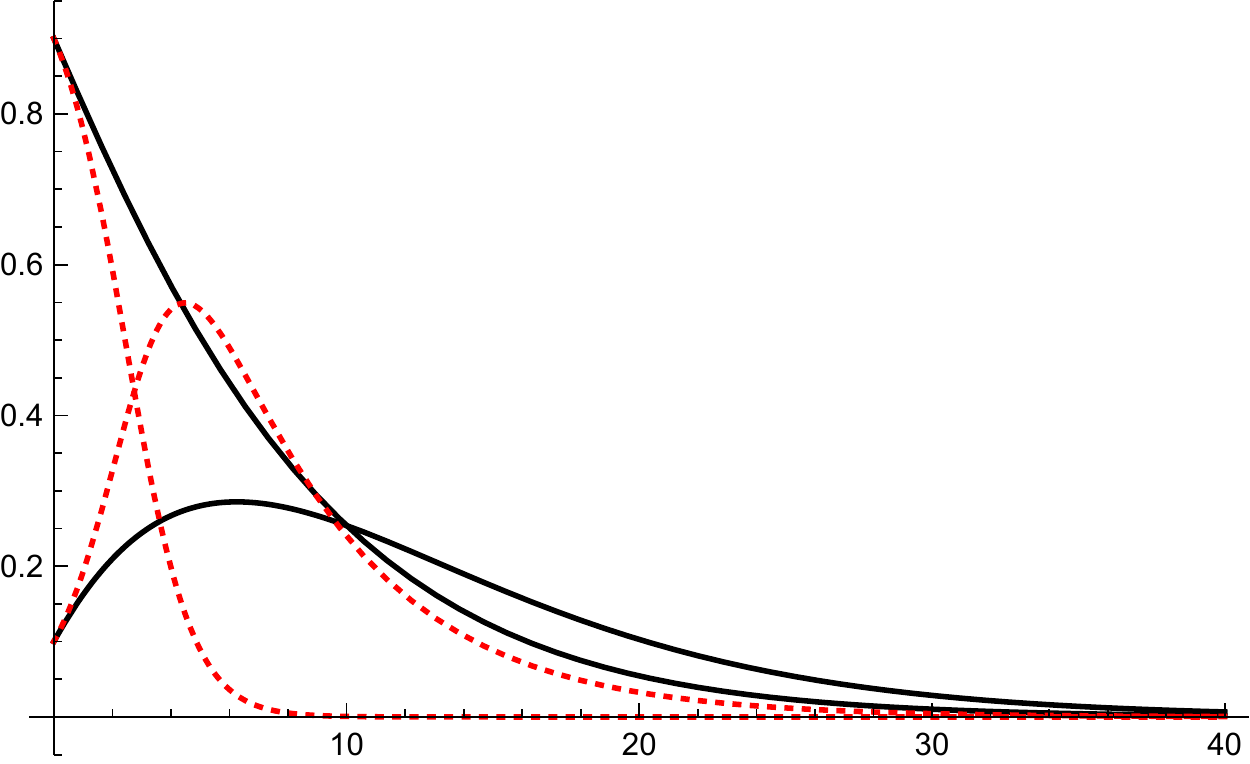} \includegraphics[width=.48\textwidth, height=.33\textwidth]{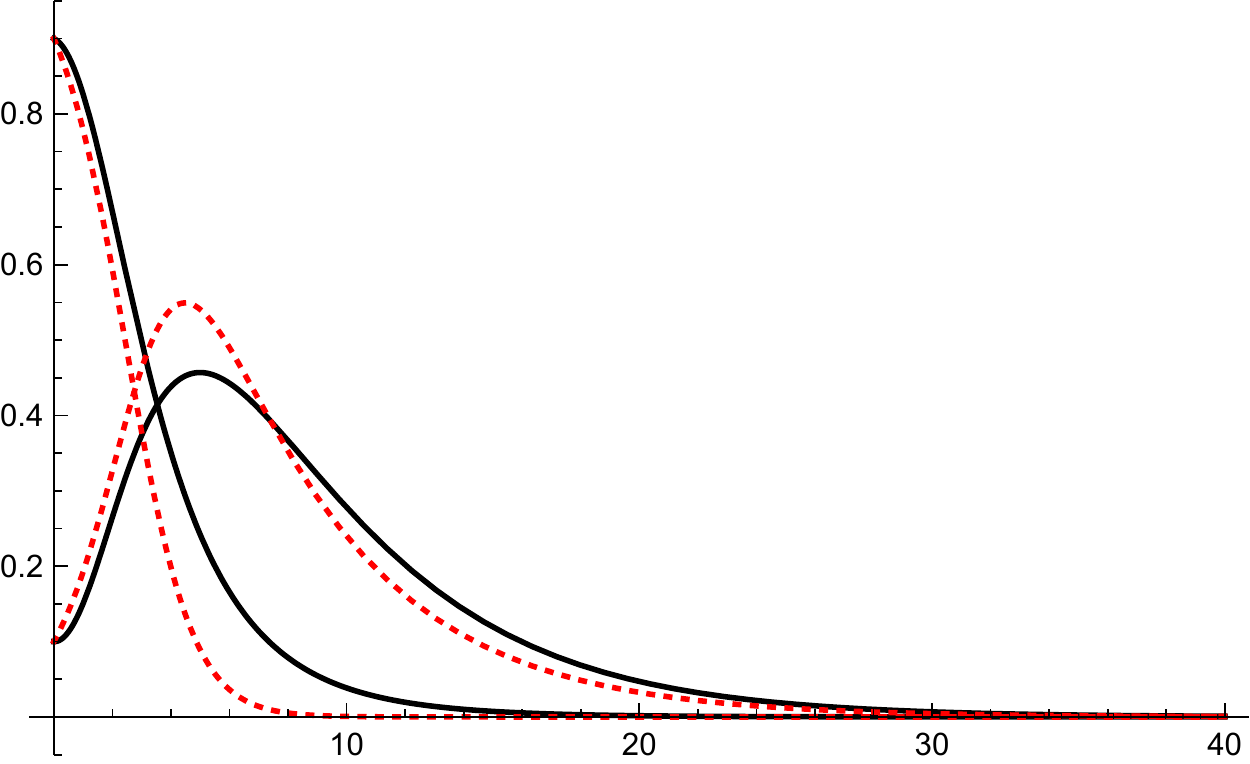} \\
\end{center}
\vspace{-0.5cm}
\caption{\small{ Left: Modified SIR with $\beta$ constant (dotted, red) and modified SIR with $\beta (t) = \alpha + \frac{1/\tau}{1+t/\tau}$ (black). Right: Modified SIR with $\beta$ constant (dotted, red) and modified SIR with $\beta (t) = \alpha + \frac{2 t/\tau^2}{1+t^2/\tau^2}$ (black). $\beta = \beta_m = 1, \alpha = 0.2$. $S_0 = 0.9$, $I_0 = 0.1$.
\label{fig:comp_traj}
}}
\end{figure}

\begin{figure}[H]
\begin{center}
\includegraphics[width=.48\textwidth, height=.33\textwidth]{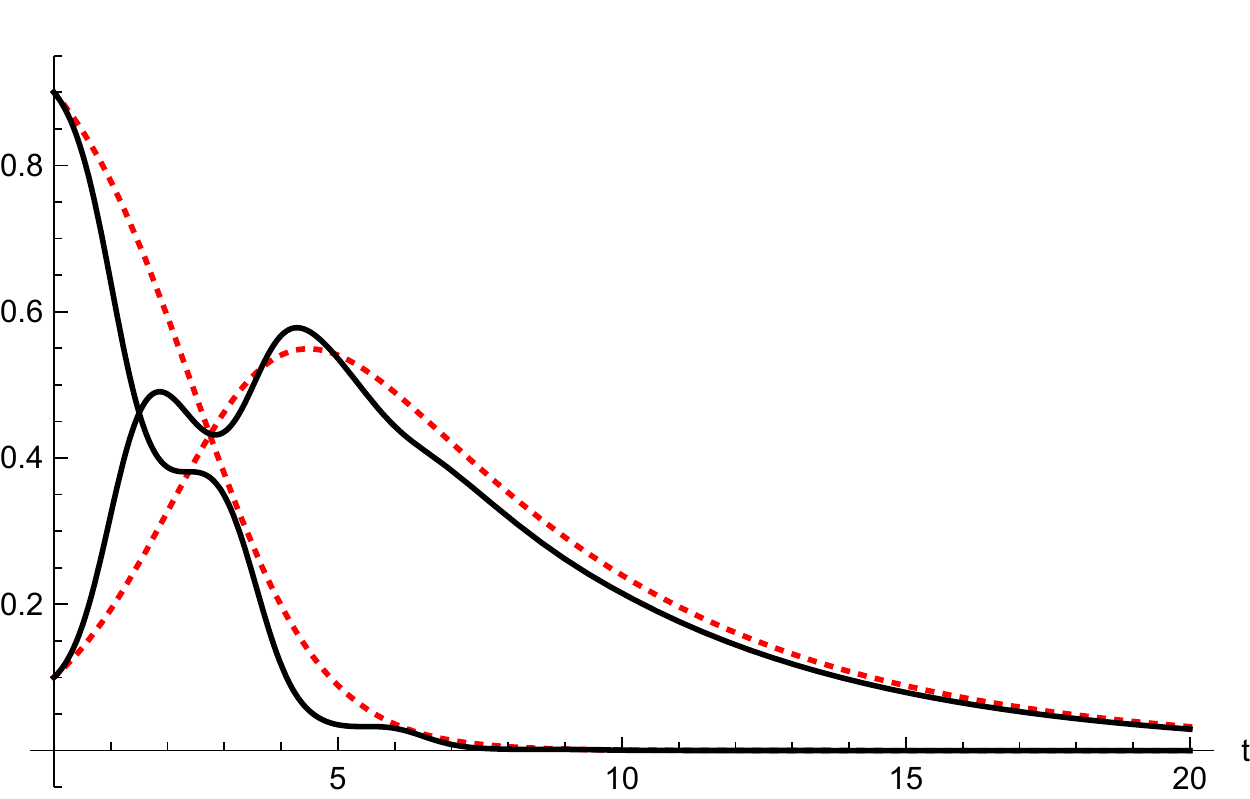} \includegraphics[width=.48\textwidth, height=.33\textwidth]{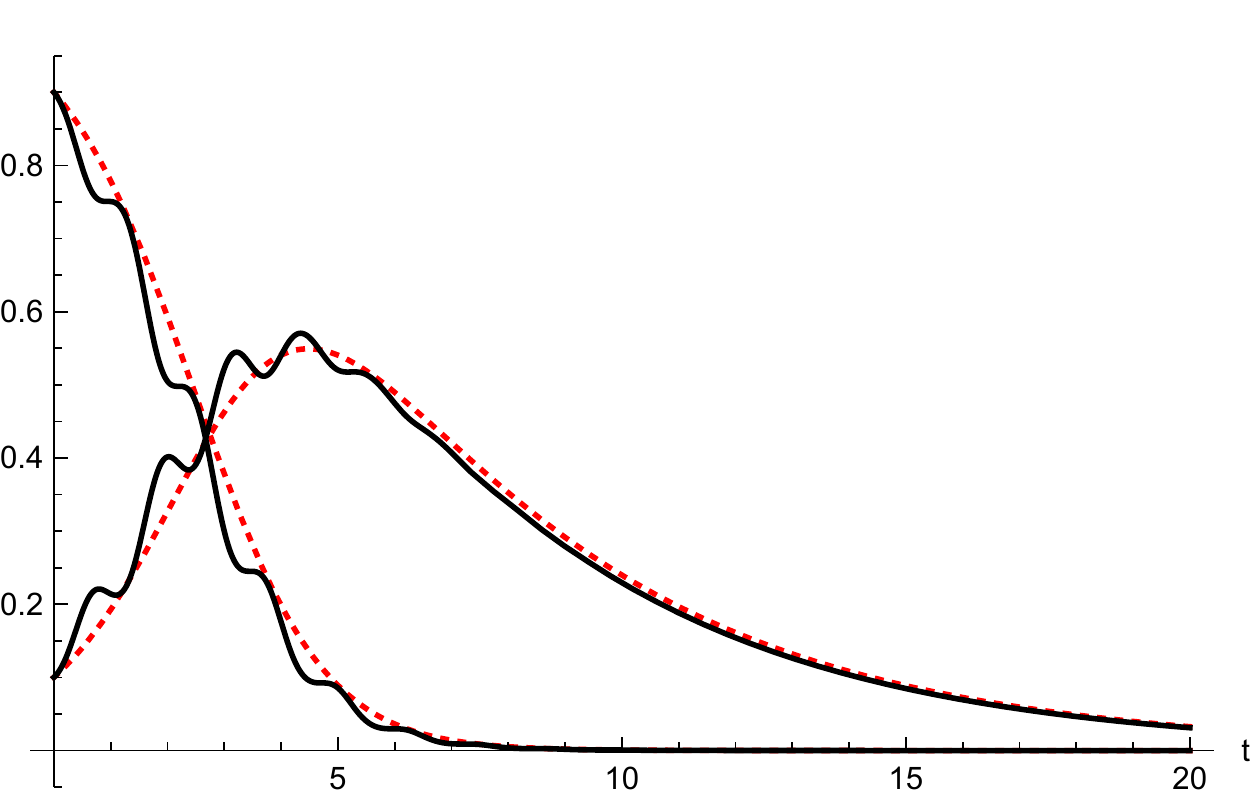} \\
\end{center}
\vspace{-0.5cm}
\caption{\small{ Modified SIR with $\beta$ constant (dotted, red) and modified SIR with $\beta (t ) = \beta_0 (1 + \gamma\, \sin (\omega\, t))$ (black) and $\gamma=1, \beta = \beta_0 = 1, \alpha = 0.2$; $S_0 = 0.9$, $I_0 = 0.1$.
Left: $\omega=2$. Right: $\omega=5$.
\label{fig:periodicbeta}
}}
\end{figure}


\subsection{A model with oscillating transmission rate}

In the case of deseases with periodic outbreaks due, for instance, to seasonal variabilities, the transmission rate function is often assumed to take a sinusoidal form (see~\cite{Keeling, Pollicott, Mummert})
\be
\beta (t ) = \beta_0 (1 + \gamma\, \sin (\omega\, t))\,,
\ee
where $0<\gamma\leq 1$ and $T=2\pi/\omega$ is the seasonality period.
In this case $y(t)$~\eqref{eq:MSIRyt} is straightforwardly obtained, but~\eqref{eq:SinMSIRydept} cannot be obtained in closed form, although its numerical integration can be easily performed and. In Figure \ref{fig:periodicbeta} two instances of periodic $\beta(t)$ are presented, where the relevance of the forcing period of $\beta(t)$ on the populations is neatly reflected (see~\cite{Mummert} for the fitting of actual influenza data with a SIR model endowed with with oscillating transmission rate, where the same type of forcing effect shown in Figure \ref{fig:periodicbeta} can be appreciated).

Summarizing, in this paper we have shown that the modified SIR model where recovered individuals are removed from the population can be exactly solved in closed form by making use of its underlying Hamiltonian structure, as well as some of its generalizations with time-dependent transmission rates. Despite SIR models are the most schematic ones in order to describe compartmental epidemiological dynamics, their features are rich enough in order to model many of their dynamical essentials. Therefore, having at hand their solutions in closed form is helpful, since this provides explicit expressions for the relevant epidemiological quantities in terms of the parameters of the model. In particular, the search for other exactly solvable models with time-dependent transmission rates would be certainly helpful in the field of non-pharmacological control strategies, and are worth to be investigated.

\section*{Acknowledgements}

This work has been partially supported by Ministerio de Ciencia e Innovaci\'on (Spain) under grants MTM2016-79639-P (AEI/FEDER, UE) and PID2019-106802GB-I00/AEI/10.13039/501100011033, and by Junta de Castilla y Le\'on (Spain) under grants BU229P18 and BU091G19. 

\small



\begin{thebibliography}{10}

\bibitem{MB2020science}
B.~F. Maier and D.~Brockmann.
\newblock {Effective containment explains subexponential growth in recent
  confirmed COVID-19 cases in China}.
\newblock {\em Science (80-. ).}, 746:eabb4557, 2020.
\newblock \href {http://dx.doi.org/10.1126/science.abb4557}
  {\path{doi:10.1126/science.abb4557}}.

\bibitem{VLABHPRYV2020}
A.~Viguerie, G.~Lorenzo, F.~Auricchio, D.~Baroli, T.~J.~R. Hughes, A.~Patton,
  A.~Reali, T.~E. Yankeelov, and A.~Veneziani.
\newblock {Simulating the spread of COVID-19 via spatially-resolved
  susceptible-exposed-infected-recovered-deceased (SEIRD) model with
  heterogeneous diffusion}.
\newblock {\em Appl. Math. Lett.}, 111:106617, 2020.
\newblock \href {http://arxiv.org/abs/2005.05320} {\path{arXiv:2005.05320}},
  \href {http://dx.doi.org/10.1016/j.aml.2020.106617}
  {\path{doi:10.1016/j.aml.2020.106617}}.

\bibitem{CG2020size}
M.~Cadoni and G.~Gaeta.
\newblock {Size and timescale of epidemics in the SIR framework}.
\newblock {\em Phys. D Nonlinear Phenom.}, 411:132626, 2020.
\newblock \href {http://arxiv.org/abs/2004.13024} {\path{arXiv:2004.13024}},
  \href {http://dx.doi.org/10.1016/j.physd.2020.132626}
  {\path{doi:10.1016/j.physd.2020.132626}}.

\bibitem{Hethcote1973}
H.~W. Hethcote.
\newblock {Asymptotic behavior in a deterministic epidemic model}.
\newblock {\em Bull. Math. Biol.}, 35(5-6):607--614, 1973.
\newblock \href {http://dx.doi.org/10.1007/BF02458365}
  {\path{doi:10.1007/BF02458365}}.

\bibitem{Bailey1975}
N.~T. Bailey.
\newblock {\em {The Mathematical Theory of Infectious Diseases}}.
\newblock Second edition, 1975.

\bibitem{Nucci2004}
M.~C. Nucci and P.~G. Leach.
\newblock {An integrable SIS model}.
\newblock {\em J. Math. Anal. Appl.}, 290(2):506--518, 2004.
\newblock \href {http://dx.doi.org/10.1016/j.jmaa.2003.10.044}
  {\path{doi:10.1016/j.jmaa.2003.10.044}}.

\bibitem{KM1927sir}
W.~O. Kermack and A.~G. McKendrick.
\newblock {A contribution to the mathematical theory of epidemics}.
\newblock {\em Proc. R. Soc. London. Ser. A, Contain. Pap. a Math. Phys.
  Character}, 115(772):700--721, 1927.
\newblock \href {http://dx.doi.org/10.1098/rspa.1927.0118}
  {\path{doi:10.1098/rspa.1927.0118}}.

\bibitem{Buonomo2020}
B.~Buonomo.
\newblock {Effects of information-dependent vaccination behavior on coronavirus
  outbreak: insights from a SIRI model}.
\newblock {\em Ric. di Mat.}, 2020.
\newblock \href {http://dx.doi.org/10.1007/s11587-020-00506-8}
  {\path{doi:10.1007/s11587-020-00506-8}}.

\bibitem{Postnikov2020sir}
E.~B. Postnikov.
\newblock {Estimation of COVID-19 dynamics `on a back-of-envelope': Does
  the simplest SIR model provide quantitative parameters and predictions?}
\newblock {\em Chaos, Solitons and Fractals}, 135, 2020.
\newblock \href {http://dx.doi.org/10.1016/j.chaos.2020.109841}
  {\path{doi:10.1016/j.chaos.2020.109841}}.

\bibitem{LA2004}
P.~G.~L. Leach and K.~Andriopoulos.
\newblock {Application of Symmetry and Symmetry Analyses to Systems of
  First-Order Equations Arising from Mathematical Modelling in Epidemiology}.
\newblock {\em Proc. Inst. Math. NAS Ukr.}, 50(1):159--169, 2004.

\bibitem{Harko2014}
T.~Harko, F.~S. Lobo, and M.~K. Mak.
\newblock {Exact analytical solutions of the Susceptible-Infected-Recovered
  (SIR) epidemic model and of the SIR model with equal death and birth rates}.
\newblock {\em Appl. Math. Comput.}, 236(0):184--194, 2014.
\newblock \href {http://arxiv.org/abs/1403.2160} {\path{arXiv:1403.2160}},
  \href {http://dx.doi.org/10.1016/j.amc.2014.03.030}
  {\path{doi:10.1016/j.amc.2014.03.030}}.

\bibitem{Barlow2020}
N.~S. Barlow and S.~J. Weinstein.
\newblock {Accurate closed-form solution of the SIR epidemic model}.
\newblock {\em Phys. D Nonlinear Phenom.}, 408:132540, 2020.
\newblock \href {http://arxiv.org/abs/2004.07833} {\path{arXiv:2004.07833}},
  \href {http://dx.doi.org/10.1016/j.physd.2020.132540}
  {\path{doi:10.1016/j.physd.2020.132540}}.

\bibitem{BBG2020hamiltonianepidemics}
A.~Ballesteros, A.~Blasco, and I.~Gutierrez-Sagredo.
\newblock {Hamiltonian structure of compartmental epidemiological models}.
\newblock {\em Phys. D Nonlinear Phenom.}, 413:132656, 2020.
\newblock \href {http://arxiv.org/abs/2006.00564} {\path{arXiv:2006.00564}},
  \href {http://dx.doi.org/10.1016/j.physd.2020.132656}
  {\path{doi:10.1016/j.physd.2020.132656}}.  
  
\bibitem{Cadoni2020}
M. Cadoni.
\newblock {How to reduce epidemic peaks keeping under control the time-span of the epidemic}.
\newblock {\em Chaos, Solitons and Fractals}, 138:109940, 2020.
\newblock \href {http://arxiv.org/abs/2004.02189} {\path{arXiv:2004.02189}},
  \href {http://dx.doi.org/10.1016/j.chaos.2020.109940}
  {\path{doi:10.1016/j.chaos.2020.109940}}.

\bibitem{Brauer1990}
F.~Brauer.
\newblock {Models for the spread of universally fatal diseases}.
\newblock {\em J. Math. Biol.}, 28(4):451--462, 1990.
\newblock \href {http://dx.doi.org/10.1007/BF00178328}
  {\path{doi:10.1007/BF00178328}}.

\bibitem{Gumral1993}
H.~G{\"{u}}mral and Y.~Nutku.
\newblock {Poisson structure of dynamical systems with three degrees of
  freedom}.
\newblock {\em J. Math. Phys.}, 34(12):5691--5723, 1993.
\newblock \href {http://dx.doi.org/10.1063/1.530278}
  {\path{doi:10.1063/1.530278}}.

\bibitem{Hethcote1976}
H. W. Hethcote. 
\newblock{Qualitative analyses of communicable disease models}.
\newblock{{\em  Math. Biosci.}, 28: 335--356, 1976}.

\bibitem{VandenDriessche2017}
P. van den Driessche.
\newblock {Reproduction numbers of infectious disease models}.
\newblock {{\em Infect. Dis. Model.}, 2(3): 288--303, 2017}.
\newblock {\href{http://dx.doi.org/10.1016/j.idm.2017.06.002}
{\path{doi:10.1016/j.idm.2017.06.002}}}.

\bibitem{Hethcote2000}
H. W. Hethcote.
\newblock {The Mathematics of Infectious Diseases}.
\newblock {{\em SIAM Review}, 42(4):599--653, 2000}.
\newblock {\href{https://doi.org/10.1137/S0036144500371907}
{\path{doi:10.1137/S0036144500371907}}}.

\bibitem{Keeling} 
M.J. Keeling, P. Rohani. 
\newblock{Modeling Infectious Diseases in Humans and Animals}.
\newblock{{\em Princeton University Press}, 2008.}

\bibitem{BacaerGomes} 
N. Bacaer, M. G. M. Gomes.
\newblock{On the Final Size of Epidemics with Seasonality}.
\newblock {{\em Bulletin of Mathematical Biology}, 71:1954--1966, 2009}. 
\newblock {\href{https://doi.org/10.1007/s11538-009-9433-7}
{\path{doi:10.1007/s11538-009-9433-7}}}.

\bibitem{Ponciano} 
J. Ponciano, M. Capistran.
\newblock{First principles modeling of nonlinear incidence rates in seasonal epidemics}.
\newblock {{\em PLoS Comput. Biol.}, 7(2):e1001079, 2011}.
\newblock {\href{https://doi.org/10.1371/journal.pcbi.1001079}
{\path{doi:10.1371/journal.pcbi.1001079}}}.

\bibitem{Liu} 
X. Liu, P. Stechlinski.
\newblock{Infectious disease models with time-varying parameters and general nonlinear incidence rate}.
\newblock {{\em Applied Mathematical Modelling}, 36:1974--1994, 2012}.
\newblock {\href{https://doi.org/10.1016/j.apm.2011.08.019}
{\path{doi:10.1016/j.apm.2011.08.019}}}.

\bibitem{Pollicott} 
M. Pollicott, H. Wang and H. Weiss.
\newblock{Extracting the time- dependent transmission rate from infection data via solution of an inverse ODE problem}.
\newblock{{\em Journal of Biological Dynamics}, 6(2):509--523, 2012}.
\newblock {\href{https://doi.org/10.1080/17513758.2011.645510}
{\path{doi:10.1080/17513758.2011.645510}}}.

\bibitem{Mummert} 
A. Mummert.
\newblock{Studying the recovery procedure for the time-dependent transmission rate(s) in epidemic models}.
\newblock{{\em J. Math. Biol.}, 67:483–507, 2013}.
\newblock {\href{https://doi.org/10.1007/s00285-012-0558-1}
{\path{doi:10.1007/s00285-012-0558-1}}}.







\end{thebibliography}
\end{document}